\documentclass[11pt]{article}
\usepackage[utf8]{inputenc}
\usepackage{amsmath}
\usepackage{amsthm}
\usepackage{fancyhdr}
\usepackage{amsfonts}
\usepackage{amssymb}
\usepackage{enumerate}
 \usepackage[normalem]{ulem}
\usepackage{tikz}
\usepackage[caption=false]{subfig}
\usepackage{pgfplots}
\usepackage{xcolor}
\usepackage{graphicx}
\usepackage{tabularx}
\usepackage[center]{caption}
\usepackage{array}
\usepackage{multirow}
\usepackage{natbib}
\usepackage{float}
        \theoremstyle{plain}
        \newtheorem{le+mma}{Lemma}[section]
        \newtheorem{proposition}{Proposition}[section]

        \theoremstyle{remark}
        \newtheorem{remark}{\bf Remark}[section]
        \theoremstyle{remark}
        
        \theoremstyle{remark}

\newcommand{\la}{\left\langle}
\newcommand{\ra}{\right\rangle}
\newcommand{\mo}{tr}
\newcommand{\ds}{\displaystyle }

\newcommand{\Kn}{{\rm Kn}}
\newcommand{\Prandtl}{{\rm Pr}}
\newcommand{\R}{{\mathbb{R}}}

\newcommand{\ent}{\mathbb{H}}
\newcommand{\entred}{{\cal H}}
\newcommand{\F}{{\bf F}}
\newcommand{\eps}{\varepsilon}

\newcommand{\tauc}{\tau_{\mathcal C}}
\newcommand{\PSigma}{P}

\newcommand{\etr}{e_{tr}}
\newcommand{\erot}{e_{rot}}
\newcommand{\evib}{e_{vib}}

\newcommand{\srot}{s_{rot}}
\newcommand{\svib}{s_{vib}}

\newcommand{\Teq}{T_{eq}}
\newcommand{\Ttr}{T_{tr}}
\newcommand{\Trot}{T_{rot}}
\newcommand{\Tvib}{T_{vib}}
\newcommand{\Ttrrel}{T_{tr}^{rel}}
\newcommand{\Trotrel}{T_{rot}^{rel}}
\newcommand{\Tvibrel}{T_{vib}^{rel}}

\newcommand{\cv}{c_v}
\newcommand{\cvtr}{c_v^{tr}}
\newcommand{\cvrot}{c_v^{rot}}
\newcommand{\cvvib}{c_v^{vib}}

\newcommand{\demi}{\frac{1}{2}}
\newcommand{\M}{{\cal M}}
\newcommand{\Mtr}{\M_{tr}}
\newcommand{\Mrot}{\M_{rot}}
\newcommand{\Mvib}{\M_{vib}}
\newcommand{\G}{{\cal G}}
\newcommand{\Gtr}{\G_{tr}}
\newcommand{\Grot}{\G_{rot}}
\newcommand{\Gvib}{\G_{vib}}

\newcommand{\Zvib}{Z_{vib}}
\newcommand{\Zrot}{Z_{rot}}
\newcommand{\thetavib}{\theta_{vib}}
\newcommand{\thetarot}{\theta_{rot}}
\newcommand{\tauvib}{\tau_{vib}}
\newcommand{\taurot}{\tau_{rot}}

\newcommand{\Trace}{{\rm Trace}}

\graphicspath{{./Figures/}}

\begin{document}

\begin{center}
{\bf An ES-BGK model for diatomic gases with correct relaxation
  rates for internal energies}
  
\vspace{1cm}
J. Mathiaud$^{1}$,  L. Mieussens$^2$, M. Pfeiffer $^3$ \\

\medskip
{\small $^1$CEA-CESTA, 15 avenue des sabli\`eres - CS 60001,
33116 Le Barp Cedex, France,\\Univ. Bordeaux, CNRS, CELIA, UMR 5107, F-33400 Talence, France.\\
{ \tt(julien.mathiaud@u-bordeaux.fr)}}\\

\medskip
{\small $^2$Univ. Bordeaux, Bordeaux INP, CNRS, IMB, UMR 5251, F-33400 Talence, France.\\
{ \tt(Luc.Mieussens@math.u-bordeaux.fr)}}

\medskip
{\small $^3$Institute of Space Systems, University of Stuttgart, Pfaffenwaldring 29, D-70569 Stuttgart,
Germany.\\
{ \tt(mpfeiffer@irs.uni-stuttgart.de))}}

\end{center}

\abstract{
We propose a new ES-BGK model for diatomic gases which allows for
translational-rotational and translational-vibrational energy
exchanges, as given by Landau-Teller and Jeans relaxation
equations. 
This model is consistent with the general definition of the vibrational and rotational collision numbers that are also commonly used in DSMC solvers.
It is proved to satisfy the H-theorem and to give
the correct transport coefficients, up to the volume viscosity.}

\tableofcontents

\section{Introduction}
In Rarefied Gas Dynamic problems, it is often useful to replace the
complicated Boltzmann equation by simplified models, both for
analytical calculations and numerical simulations. These model
equations describe intermolecular collisions by drift and diffusion in
the velocity space (Fokker-Plank models,
see~\cite{cercignani,Grj2011,Mathiaud2016}) or by relaxation to a
local equilibrium: this latter approach, first proposed by Bathnagar
et al.~\cite{bgk} and Welander~\cite{welander} leads to the so called
BGK equation.

The BGK equation describes the evolution of a rarefied monoatomic
flow, and is designed to satisfy several properties of the Boltzmann
equation, like conservation laws, H-theorem, and correct shear
viscosity coefficient in the compressible Navier-Stokes asymptotics as
obtained by the Chapman-Enskog expansion. However, the BGK equation
contains a single free parameter (the relaxation time) which is not
sufficient to independently fit the correct value of the heat transfer
coefficient, which leads to a constant Prandtl number equal to
1.

Several modifications of the BGK equation have been proposed to fix
this problem, like the Ellipsoidal-Statistical (ES-BGK)~\cite{holway}
and the Shakhov~\cite{S_model} models that are
the most popular (see~\cite{struchtrup,liu} for other models). While
the ES-BGK model was already extended to polyatomic gases
in~\cite{holway}, this is the more recent extension of Andriès et
al.~\cite{ALPP} to polyatomic gases with rotational energy that is
mostly used in the litterature~\cite{ABLP,KKA_2019}. In the same paper, the authors also proved for the first
time that the ES-BGK model satisfies the H-theorem. While the Shakhov
model has also been extended to polyatomic
gases~\cite{R_model,WYLX_2017}, it cannot satisfy the H-theorem, since
it is a perturbative model in which the distribution function can take
negative values. Several extensions of the BGK equation have
also been proposed for discrete internal energy
levels~\cite{Morse_1964,bisi2016,DMM_2021} or thermally perfect gases~\cite{MM2021}.

In~\cite{DMM_2021}, the approach of Andriès et al.~\cite{ALPP} was
applied to extend the ES-BGK model to diatomic gases in which a
discrete vibrational energy is taken into account. This model was
designed to obtain the correct Prandtl number, as well as the correct
relaxation times of internal energies, as defined by Landau-Teller and
Jeans equations. However, first simulations~\cite{BBDMM_2022} show some
discrepancies with DSMC simulations, especially for the rotational and
vibrational temperature profiles, which suggests that energy exchanges
are not taken into account in the same way in the ES-BGK model and in
the DSMC solver.

Recently, Pfeiffer~\cite{pfeiffer2018} proposed an ES-BGK based
particle simulation of diatomic rarefied flows in which he proposed
a specific treatment of internal energy exchanges. His results
show very good agreement with DSMC. However,
the algorithm used in~\cite{pfeiffer2018} is not derived from a complete
kinetic model.

In this paper, we propose an ES-BGK model which is consistent with
the numerical method of~\cite{pfeiffer2018}, and based on the
theoretical framework of~\cite{DMM_2021}. The main modifications with
respect to the model of~\cite{DMM_2021} are the following
ingredients, taken from~\cite{pfeiffer2018}:
\begin{enumerate}
\item The energy relaxation time scale is proportional to the mean
  collision time $\tauc$ rather than to the relaxation time $\tau$, as
  opposed to what is done in~\cite{DMM_2021}. This corresponds to the common definition of the relaxation time of inner degrees of freedom and the associated definition of the vibrational and rotational collision numbers~\cite{parker1959,BS_2017,Haas1994}. 
\item The Landau-Teller and Jeans equations are used to define
  translational-rotational and translational-vibrational energy
  exchanges, and induce a relaxation of rotational and vibrational
  temperatures to the {\it translational} temperature, as described and discussed in detail in Haas et al.~\cite{Haas1994} for the DSMC method. Again,
  this is different to what is done in~\cite{DMM_2021}, where the
  model induces a relaxation of internal temperatures to the overall
  temperature.
\end{enumerate}
Numerical tests in space homogeneous cases illustrate the excellent
agreement between our new model and DSMC simulations.

Note that this new approach can also be used even if the vibration
modes are not taken into account: we obtain an ES-BGK model for
diatomic gases in rotational non-equilibrium which is different from
the ES-BGK model of Andriès et al.~\cite{ALPP}. However, both models are proved to
be equivalent up to a correction factor of the relaxation time, or
equivalently of the collision number $Z_{rot}$, but this correction
factor can be quite large.

Moreover, our new ES-BGK model is proved to satisfy the H-theorem,
with a proof that is more involved than that for~\cite{DMM_2021}. A
Chapman-Enskog expansion gives the corresponding transport
coefficients, and we obtain the following strong result: the volume
viscosity is shown to be the same as that obtained in the Boltzmann
equation with two fast and slow energy modes.

The outline of our paper is the following. Sections~\ref{sec:thermo}
and~\ref{sec:dist_funct} are devoted to the definition of internal
energies and temperature, relaxation times, and distribution
functions. Our ES-BGK model is derived and analyzed in Section~\ref{sec:
  esbgk}. The results of Chapman-Enskog expansion is given in
Section~\ref{sec: chap}. A reduced model is proposed in
Section~\ref{sec: r_esbgk} to reduce its computational
complexity. Finally, the properties of our model are illustrated by
some numerical results in section~\ref{sec: num}.

\section{Internal energies of  diatomic gases }
\label{sec:thermo} 

\subsection{The different macroscopic internal energies at equilibrium}

In this paper we consider diatomic perfect gases for which each
molecule has several degrees of freedom: translation, rotation and
vibration. At the macroscopic level, a gas in thermodynamical equilibrium
at temperature $T$ has different specific energies associated to each
mode. For translational, rotational and vibrational (in case of the harmonic oscillator
model) modes, the corresponding specific energies are 
\begin{equation} \label{eq-etr_erot} 
e_{tr}(T)=\frac32 RT, \quad \quad
e_{rot}(T)=\frac{\delta}2 RT, \quad  \quad
     e_{vib}(T)=\frac{RT_0}{\exp\left(T_0/T\right)-1}, 
\end{equation}
where the specific total energy is
\begin{equation}  \label{eq-e}
e(T) = e_{tr}(T)+e_{rot}(T)+e_{vib}(T).
\end{equation}
Here, $\delta=2$ is the number of degrees of freedom of rotation, $R$
is the gas constant per unit mass and $T_0$ is the characteristic
vibrational temperature.

Note that polyatomic molecules could be considered here with $\delta>2$ and a vibration energy as given by a sum over all harmonic oscillators of the molecule~\cite{pfeiffer2019extension}. However, this would change some details in our mathematical proofs, so that an extension of our approach to polyatomic molecules is deferred to future work.

\subsection{Mathematical properties of the energy functions}
\label{subsec:prop_energies}

For each energy mode, a temperature can be defined as follows. We denote by $e_i^{-1}$ the function that maps any given energy $E$ to
the corresponding temperature. That is to say the temperature $T$
corresponding to a given energy $E$ is such that $e_{\alpha}(T) = E$, where
$\alpha$ stands for $tr$, $rot$, and $vib$, and is denoted by $T=
e_{\alpha}^{-1}(E)$. Simple computations give
\begin{equation}\label{eq-inv_etrrotvib} 
  e_{tr}^{-1}(E) = \frac{2}{3R}E, \qquad  
 e_{rot}^{-1}(E) =  \frac{2}{\delta R}E, \qquad 
 e_{vib}^{-1}(E) = {T_0}/ {\log\left( 1+\frac{RT_0}{E}
   \right)}.
\end{equation}
The total energy function, which is clearly invertible, cannot be
inverted analytically, and we simply set
\begin{equation}\label{eq-inv_e} 
 T=  e^{-1}(E) \quad \text{ such that } \quad E = \frac{3+\delta}{2}RT  + \frac{RT_0}{\exp\left(T_0/T\right)-1}.
\end{equation}

For each energy mode, we can also define a specific heat at constant
volume $c_v^{\alpha}(T)  = \frac{d e_{\alpha}(T)}{dT}$. For
translational and rotational energies, the specific heats are constant:
\begin{equation}  \label{eq-cvtrrot}
\cvtr = \frac{3}{2}R, \qquad \cvrot = \frac{\delta}{2}R,
\end{equation}
while for vibrational energy, we find
\begin{equation}  \label{eq-icvib}
 \cvvib(T)= R\frac{T_0^2}{T^2}\frac{\exp(T_0/T)} {\left( \exp(T_0/T)-1 \right)^2}.
\end{equation}
Note that $\cvvib$ can be proved to be an increasing function of $T$
which is bounded by $R$. This also implies that $\evib$ is a convex
function (see appendix~\ref{app:evibconvex}).

Finally, we also define for each mode a specific entropy $s_{\alpha}$
such that $\frac{ds_{\alpha}(E)}{dE} =
\frac{1}{e^{-1}_\alpha(E)}$. This gives, up to any arbitrary constant
\begin{equation}  \label{eq-entropies}
s_{tr}(E) =  \frac32 R \log E, \quad s_{rot}(E) =  \frac{\delta}{2} R
\log E, \quad
s_{vib}(E) = R \Bigl( 
  \log (1+\frac{E}{RT_0}) 
  + \frac{E}{RT_0}\log(1+\frac{RT_0}{E})
\Bigr),
\end{equation}
and we define the total entropy (at constant density)
\begin{equation}\label{eq-S} 
  {\cal S}(E_1,E_2,E_3) = s_{tr}(E_1) + s_{rot}(E_2) + s_{vib}(E_3).
  \end{equation}

\section{Distribution functions, moments, and temperatures}
\label{sec:dist_funct} 

\subsection{Distribution function}

The state of any gas molecule is described by its position $x$,
its velocity $v$, its rotational energy $\varepsilon$, and its
discrete vibrational energy $iRT_0$, where $i$ is the
$i$th vibrational energy level and $T_0$ is the characteristic
vibrational temperature of the gas, in the case of the usual simple harmonic
oscillator model.

The distribution function of the gas is the mass density
$f(t,x,v,\varepsilon,i)$ of molecules that at time
$t$ are located in a elementary volume $dx$ centered in $x$, have the
velocity $v$ in a elementary volume $dv$, have the rotational energy
$\varepsilon$ centered in $d\varepsilon$ and the discrete vibrational
energy $iRT_0$. The macroscopic densities of mass $\rho$, momentum $\rho u$, and
internal energy $\rho E$ are
\begin{equation}\label{eq-mtsf} 
\rho=\la f \ra_{v,\varepsilon,i},\qquad \rho u =\la v f \ra_{v,\varepsilon,i},\qquad \rho E(f)=\la \left(\frac{1}{2}|v-u|^2+\varepsilon+iRT_0\right)f \ra_{v,\varepsilon,i}.
\end{equation}
The dependence of $E$ on $f$
is intentionally made explicit, and we denote by $\la \phi
\ra_{v,\varepsilon,i}(t,x)=\sum_{i=0}^{+\infty}\int_{\mathbb{R}^3}\int_{\mathbb{R}}\phi(t,x,v,\varepsilon,i)d\varepsilon
dv$ the integral of any function $\phi$.
The specific internal energy $E(f)$ can be decomposed
into
\begin{equation}  \label{eq-EEtrErotEvib}
E(f) = E_{tr}(f) +  E_{rot}(f) + E_{vib}(f),
\end{equation}
where specific energies $E_{tr}(f)$, $E_{rot}(f)$  and $E_{vib}(f)$ are respectively  associated with  translational motion of
particles,  rotational mode and  vibrational mode through:
\begin{equation}\label{eq-Etrrotvib} 
\rho E_{tr}(f)=\la \frac{1}{2}|v-u|^2 f \ra_{v,\varepsilon,i},\qquad 
\rho E_{rot}(f)=\la \varepsilon f \ra_{v,\varepsilon,i},\qquad
\rho E_{vib}(f)=\la iRT_0 f \ra_{v,\varepsilon,i}.
\end{equation}

We also define the pressure tensor $\PSigma(f)$ and the heat flux $q(f)$ by
\begin{equation}\label{eq-Theta_q} 
  \PSigma(f) =\la (v-u)\otimes (v-u) f \ra_{v,\varepsilon,i} \quad,
  \quad 
  q(f)=\la \left(\frac{1}{2}|v-u|^2 + \varepsilon + iRT_0\right) (v-u) f \ra_{v,\varepsilon,i}
\end{equation}
and we denote by $\Theta$ the tensor such that $\PSigma(f) =
  \rho \Theta$. 

\subsection{Internal temperatures}
\label{subsec:int_temp}

For a given distribution function $f$, the translational, rotational, and
vibrational temperatures are defined by
\begin{equation}\label{eq-T_tr_rot_vib} 
  T_{tr} = e_{tr}^{-1}(E_{tr}(f)), \quad 
  T_{rot} = e_{rot}^{-1}(E_{rot}(f)), \quad 
  T_{vib} = e_{vib}^{-1}(E_{vib}(f)), 
\end{equation}
so that we have
\begin{equation}\label{eq-E_T} 
E_{tr}(f)=\etr(\Ttr) = \frac{3}{2}RT_{tr},\quad 
E_{rot}(f)=\erot(\Trot) = \frac{\delta}{2}RT_{rot},\quad 
E_{vib}(f)=\evib(\Tvib) = \frac{RT_0}{\exp(T_0/T_{vib})-1},
\end{equation}
see section~\ref{sec:thermo}. A number
of degrees of freedom $\delta_v(T_{vib})$ for the vibration mode can
be defined such that $\ds E_{vib}(f) =
\frac{\delta_v(T_{vib})}{2}RT_{vib}$, which leads to
\begin{equation}\label{eq-deltav} 
 \delta_v(T_{vib})=\frac{2T_0/T_{vib}}{\exp(T_0/T_{vib})-1}.
\end{equation}
This number is not an integer, is temperature dependent, and tends to $2$ for large
$T_{vib}$.

The overall or equilibrium temperature $T_{eq}$  is the temperature corresponding
to the total internal energy, that is to say 
\begin{equation}\label{eq-Teq} 
  T_{eq}  = e^{-1}(E(f)),
\end{equation}
and $\Teq$ can be obtained by numerically solving
\begin{equation}\label{eq-ETeq} 
  E(f) = \frac{3+\delta}{2}RT_{eq}  + \frac{RT_0}{\exp\left(T_0/T_{eq}\right)-1}.
\end{equation}

\subsection{Macroscopic relaxation phenomena \label{sec: energ}}
The common description of the relaxation of internal energies with Jeans and
Landau-Teller equations~\cite{parker1959} as also typically used in DSMC codes (see~\cite{Haas1994,BS_2017,pfeiffer2018}) is given as:
\begin{align} 
\frac{d}{dt}
&  e_{rot}(T_{rot})=\frac{1}{Z_{rot}\tauc}(e_{rot}(T_{tr})-e_{rot}(T_{rot})),
\label{eq-LTerot}    \\
\frac{d}{dt}
&  e_{vib}(T_{vib})=\frac{1}{Z_{vib}\tauc}(e_{vib}(T_{tr})-e_{vib}(T_{vib})),\label{eq-LTevib}
\end{align}
 where $Z_{rot}$ and $Z_{vib}$ are the mean number of
 collisions necessary to have an exchange of rotational and vibrational
 energy, respectively, and $\tauc$ is a characteristic time of collision (see
 section~\ref{subsec:chartime}). The equation for translational energy is
\begin{equation}\label{eq-LTetr} 
\frac{d}{dt} e_{tr}(T_{tr})=-\frac{(e_{rot}(T_{tr})-e_{rot}(T_{rot}))}{Z_{rot}\tauc}-\frac{e_{vib}(T_{tr})-e_{vib}(T_{vib})}{Z_{vib}\tauc}.
\end{equation} 
which is deduced from the conservation of total energy. Our ES-BGK
model will be designed to satisfy these relaxation equations.

\begin{remark}
  These equations are different from those used
  in~\cite{DMM_2021}. Indeed, first, they induce a relaxation of
  $\Trot$ and $\Tvib$ to $\Ttr$, while a relaxation to the overall
  temperature $\Teq$ was imposed in~\cite{DMM_2021}, and second the
  relaxation time used here is the collision time $\tauc$, while the ES-BGK
  relaxation time $\tau$ was used in~\cite{DMM_2021} (see
  sections~\ref{subsec:chartime} and~\ref{subsec:compAndries}). The use of $\tauc$ and the relaxation to translation temperature $\Ttr$ instead of equilibrium temperature $\Teq$ corresponds to the most frequently used definition of the Landau-Teller and Jeans equation. A detailed discussion of the relaxation to the translation temperature instead of the equilibrium temperature can be found in Haas et al.~\cite{Haas1994}.
\end{remark}

\subsection{Some remarks on  the  collision time  $\tauc$} \label{subsec:chartime}

The collision numbers $Z_{rot}$ and $Z_{vib}$ describe the average required number of collisions of the gas during which it undergoes a relaxation process in the rotational and vibratory degrees of freedom, respectively. 
Therefore, the characteristic time $\tauc$ should be chosen at the
mean collision time of the gas \cite{pfeiffer2018} and is generally
not equal to the relaxation time of the ES-BGK model which is chosen to
represent the correct viscosity (as opposed to what is done
in~\cite{ALPP} and~\cite{DMM_2021}).
The difference between the relaxation time $\tau$ and the collision time $\tauc$ depends on the molecular model used. For example, if we look at the Variable Soft Sphere model (VSS) often used in DSMC, the collision time is given by~\cite{bird}:
\begin{equation}\label{eq-taucvss} 
\tauc^{VSS} = \frac{\alpha(5-2\omega)(7-2\omega)}{5(\alpha+1)(\alpha+2)}\frac{\mu}{p}= \frac{\alpha(5-2\omega)(7-2\omega)}{5(\alpha+1)(\alpha+2)}\tau Pr,    
\end{equation}
with $\alpha$ the diffusion factor of the VSS model, $\omega$ the
exponential factor of the temperature dependency in the viscosity, and
$\Pr$ is the Prandtl number. Here, we have used the usual relation
$\tau = \mu / (p\Prandtl)$ for ES-BGK, which will be proved below. The variable hard sphere (VHS) model can simply be achieved by setting $\alpha=1$ which gives:
\begin{equation}
\tauc^{VHS} = \frac{(5-2\omega)(7-2\omega)}{30}\frac{\mu}{p}= \frac{(5-2\omega)(7-2\omega)}{30}\tau Pr.
\end{equation}
And finally we get the HS model from it when $\omega=0.5$:
\begin{equation}
\tauc^{HS} = \frac{4}{5}\frac{\mu}{p} = \frac{4}{5}\tau Pr.
\end{equation}

\section{\label{sec: esbgk}ES-BGK model ant its mathematical properties }

\subsection{Construction of the model}\label{subsec:construct}

The evolution equation for $f$ is the Boltzmann equation
\begin{equation}
\label{eq: cinetique}
\partial_t f+v \cdot\nabla f=Q(f),
\end{equation}
where $Q(f)$ is the collision operator
(see~\cite{Giovangigli1999}). The corresponding local Maxwellian equilibrium in
velocity and energy is defined by
\begin{equation}\label{eq-Mf} 
\mathcal{M}[f](v,\varepsilon,i)= \mathcal{M}_{\mo}[f](v)\mathcal{M}_{rot}[f](\varepsilon)\mathcal{M}_{vib}[f](i),
\end{equation}
with
\begin{align*}
& \mathcal{M}_{\mo}[f](v)=\frac{\rho}{(2\pi
                 RT_{eq})^{3/2}}\exp\left(-\frac{|v-u|^2}{2RT_{eq}}
                 \right), \qquad 
 \mathcal{M}_{rot}[f](\varepsilon)=\frac{\Lambda(\delta)\varepsilon^{\frac{\delta-2}{2}}}{(R T_{eq})^{\delta/2}}\exp\left( -\frac{\varepsilon}{R T_{eq}} \right),
\\
& {M}_{vib}[f](i)=\left(1-\exp(-T_0/T_{eq})\right)\exp\left(-i\frac{T_0}{T_{eq}} \right),
\end{align*}
where $\Lambda(\delta) = 1/\Gamma(\frac{\delta}{2})$, with
$\Gamma$ the usual gamma function.

This Maxwellian distribution can be used to define the BGK
approximation~\cite{Mathiaud2019}, where $Q(f)$ is replaced by
$\frac{1}{\tau}(\mathcal{M}[f] - f) $, where $\tau$ is a relaxation
time. This approximation has the same conservation and entropy
properties as the original Boltzmann operator, but is simpler for deterministic
numerical simulations. However, the single relaxation time cannot account
for the various time scales of the original
problem. Indeed, such a model gives the same value for rotational and
vibrational relaxation times, and the same value for relaxation times
of viscous and thermal fluxes, leading to the usual incorrect Prandtl number $\Prandtl =
1$.

Additional relaxation times can be added in this model by using the
ES-BGK approach exposed in~\cite{DMM_2021}: the idea is to modify the
equilibrium temperature $\Teq$ in $\Mtr$, $\Mrot$, and $\Mvib$ so as
to obtain the correct relaxation times. Indeed, our ES-BGK collision
operator is
\begin{equation}
\label{eq: operator} 
Q(f)=\frac{1}{\tau}(\G[f]-f),
\end{equation}
with $\G[f](v,\eps,i) = \Gtr[f](v)\Grot[f](\eps)\Gvib[f](i)$, where
\begin{equation}\label{eq-Gaussalpha} 
\begin{split}
& \Gtr[f](v) =\frac{\rho}{  \sqrt{\det(2\pi \Pi)}}  \exp\left(-\frac{1}{2} (v-u)^{T}\,\Pi^{-1}\,(v-u))\right),\\
& \Grot[f](\varepsilon)=\frac{\Lambda(\delta)}{(R
  {T_{rot}^{rel}})^{\delta/2}}\varepsilon^{\frac{\delta-2}{2}}\exp\left(
  -\frac{\varepsilon}{R T_{rot}^{rel}}\right), \\
& \Gvib[f](i)=(1-\exp(-T_0/T_{vib}^{rel})) \exp\left(-i\frac{T_0}{T_{vib}^{rel}}\right),
\end{split}
\end{equation}
are distributions associated to the energies of
translation, rotation and vibration of the molecules. The relaxation
tensor $\Pi$ and temperatures $\Trotrel$ and $\Tvibrel$ are defined as
follows.

First, note the following integral properties
\begin{align} 
&  \int_{\R^3} \Gtr[f](v) \, dv = \rho,  \quad
    \int_{\R^3}v \Gtr[f](v)\, dv = \rho u, \quad
    \int_{\R^3}(v-u) \otimes (v-u) \Gtr[f](v)\, dv  = \rho \Pi
                 \label{eq-intGtr} \\              
  & \int_0^{+\infty} \Grot[f](\eps) \, d\eps = 1, \qquad
    \int_0^{+\infty} \eps \Grot[f](\eps) \, d\eps =\erot(\Trotrel),\label{eq-intGrot} \\
  & \sum_{i=0}^{+\infty} \Gvib[f](i) = 1, \qquad
    \sum_{i=0}^{+\infty} iRT_0 \Gvib[f](i) = \evib(\Tvibrel). \label{eq-intGvib}
\end{align}

Now, $\Trotrel$ and $\Tvibrel$ are defined so that our ES-BGK
model~\eqref{eq: cinetique}--\eqref{eq-Gaussalpha} satisfies (in the
space homogeneous case) the Landau-Teller and Jeans
equations~\eqref{eq-LTerot}--\eqref{eq-LTevib}. This gives
\begin{align}
  \erot(\Trotrel) & =  \erot(\Trot) + \frac{\tau}{\Zrot \tauc}
                    (\erot(\Ttr) - \erot(\Trot)),\label{eq-Trotrel}
  \\
  \evib(\Tvibrel) & = \evib(\Tvib) + \frac{\tau}{\Zvib \tauc}
                    (\evib(\Ttr) - \evib(\Tvib)),\label{eq-Tvibrel}
  \end{align}

We also need a relaxation translational temperature $\Ttrrel$, defined by
\begin{equation}  \label{eq-Ttrrel}
\rho \etr(\Ttrrel)  = \int_{\R^3}\demi |v-u|^2 \Gtr[f](v)\, dv 
\end{equation}
which reads $\etr(\Ttrrel) = \demi \Trace(\Pi)$, or equivalently
$\Ttrrel = \frac{1}{3R}\Trace(\Pi)$. Then, the conservation of total
energy of our model requires
$\la (\demi |v-u|^2 +\eps + iRT_0) \G[f] \ra_{v,\varepsilon,i} = \rho
E(f)$, which gives the following definition of $\Ttrrel$:
\begin{equation}  \label{eq-Ttrrel2}
  \etr(\Ttrrel) = \etr(\Ttr)
  -\frac{\tau}{\Zrot \tauc} (\erot(\Ttr) - \erot(\Trot))
  -\frac{\tau}{\Zvib \tauc} (\evib(\Ttr) - \evib(\Tvib))  .
\end{equation}

Now, the relaxation tensor $\Pi$ is defined as follows. In the
homogeneous case, our ES-BGK model makes the heat flux relax exponentially fast to 0 with
relaxation time $\tau$. We impose that the deviation
of $\Theta$ to its trace value $R\Ttr I$ relaxes to zero too,
with relaxation time $\tau \Pr$. This gives
\begin{equation}  \label{eq-defPi}
\Pi=
 R \Ttrrel I 
+ \frac{\Prandtl-1}{\Prandtl} ( \Theta - R \Ttr I ).
  \end{equation}

The relaxation time $\tau$ is defined so that our ES-BGK
model is consistent with the compressible Navier-Stokes equations with
shear viscosity $\mu$ (see section~\ref{sec: chap}): this gives
\begin{equation}\label{eq-deftau} 
  \tau = \frac{\mu} {\rho R \Ttr \Prandtl}.
\end{equation}
Moreover, note that a temperature power law dependence of $\mu$ is generally
chosen, which is related to the intermolecular collision model of the
Boltzmann equation (see~\cite{bird} for instance).

Finally, note that collision numbers $\Zrot$ and $\Zvib$ might be temperature
dependant (models of Parker and Millikan-White): in this case,
they have to be defined at the
translational temperature $\Ttr$. In the same way, $\mu$ should also
be defined  at $\Ttr$ in~\eqref{eq-deftau}, so that $\tau$ depends on
$\Ttr$, like $\tauc$. However, to make notations simpler, the
dependence on $\Ttr$ of $\Zrot$, $\Zvib$, $\tau$, and $\tauc$ is not
made explicit in the remaining of this paper.

\subsection{Definition of the model}

Our model is not always well defined: indeed, it requires that the
relaxation energies are positive, and that the relaxation tensor $\Pi$
is positive definite. These constraints are analyzed in the following
two propositions, where it is shown that they depend on the
translational temperature only via values of $\Zrot$, $\Zvib$, $\tau$,
$\tauc$, and $\cvvib$.

\begin{proposition}[Positiveness of relaxation energies] \label{prop:pos_ener}
  For positive $\Ttr$, $\Trot$ and $\Tvib$, the relaxation energies defined by~\eqref{eq-Trotrel}, \eqref{eq-Tvibrel},
  and~\eqref{eq-Ttrrel2}, are positive if
  \begin{equation}\label{eq-erelpositive} 
    \frac{\tau}{\Zrot\tauc}<1, \qquad 
    \frac{\tau}{\Zvib\tauc}<1, \qquad \text{ and } \qquad 
    \frac{\tau}{\Zrot\tauc}\frac{\cvrot}{\cvtr}
    +  \frac{\tau}{\Zvib\tauc} \frac{\cvvib(\Ttr)}{\cvtr}   < 1.
  \end{equation}
\end{proposition}

\begin{proof}
  The positivity of $\erot(\Trotrel)$ and $\evib(\Tvibrel)$ is
  obtained by writing relations~\eqref{eq-Trotrel} and
  \eqref{eq-Tvibrel} as linear combinations that are clearly strictly convex under the
  necessary and sufficient conditions $\frac{\tau}{\Zrot \tauc}<1$ and
  $\frac{\tau}{\Zvib \tauc}<1$.

For $\etr(\Ttrrel)$, we rewrite~\eqref{eq-Ttrrel2} as
\begin{equation*}
 \Ttrrel = \Ttr
 -\frac{\tau}{\Zrot \tauc} \frac{\cvrot}{\cvtr}(\Ttr - \Trot)
 -\frac{\tau}{\Zvib \tauc} \frac{1}{\cvtr} (\evib(\Ttr) - \evib(\Tvib)),
\end{equation*}
see~\eqref{eq-E_T} and~\eqref{eq-cvtrrot}. Then we use the mean value
theorem applied to the function $\evib$ to get
\begin{equation} \label{eq-cond1}
 \Ttrrel = \Ttr
 -\frac{\tau}{\Zrot \tauc} \frac{\cvrot}{\cvtr}(\Ttr - \Trot)
 -\frac{\tau}{\Zvib \tauc} \frac{\cvvib(T_1)}{\cvtr} (\Ttr - \Tvib),
\end{equation}
where $T_1$ lies between $\Ttr$ and $\Tvib$ and is such that
$\evib(\Ttr) - \evib(\Tvib) = \cvvib(T_1) (\Ttr-\Tvib)$, and we remind we have
used $\cvvib(T) =d\evib(T) / dT $.

Now, for the positiveness of $ \Ttrrel $, the most restrictive case is
when $\Ttr-\Trot\geq 0$ and $\Ttr-\Tvib\geq 0$, that we assume now.
Moreover, the positiveness of $ \Ttrrel $
and hence of $\etr(\Ttrrel)$, is obtained by writing~\eqref{eq-cond1}
as a linear combination of $\Ttr$, $\Trot$, and $\Tvib$ which
is strictly convex under the condition
\begin{equation}\label{eq-condetrelpos} 
  \frac{\tau}{\Zrot \tauc} \frac{\cvrot}{\cvtr}+\frac{\tau}{\Zvib
    \tauc} \frac{\cvvib(T_1)}{\cvtr}  <1 .
\end{equation}
Since $\cvvib$ is an increasing function (see
section~\ref{subsec:prop_energies}), and since we assumed
$\Ttr\geq\Tvib$, therefore $\cvvib(T_1)\leq \cvvib(\Ttr)$, which gives
the last condition of~\eqref{eq-erelpositive}.

For the other cases, it can easily be proved that this condition is
sufficient too: in the case ($\Ttr-\Trot\leq 0$ and
$\Ttr-\Tvib\leq 0$),~\eqref{eq-cond1} is always true, and in the cases
($\Ttr-\Trot\geq 0$ and $\Ttr-\Tvib\leq 0$) and ($\Ttr-\Trot\leq 0$
and $\Ttr-\Tvib\geq 0$), ~\eqref{eq-cond1} is true under conditions
$1
-\frac{\tau}{\Zrot\tauc}\frac{c_v^{rot}}{c_v^{tr}}\geq 0$ and
$1 -
\frac{\tau}{\Zvib\tauc}\frac{c_v^{vib}(\Ttr)}{c_v^{tr}} \geq 0$,
respectively.

\end{proof}

\begin{proposition}[Positiveness of tensor $\Pi$] \label{pipositive}
Let $\Ttr$, $\Trot$ and $\Tvib$ be three positive temperatures, and a Prandtl number
$\frac23<\Prandtl\leq1$. We assume~\eqref{eq-erelpositive} holds, then the tensor $\Pi$ defined by~\eqref{eq-defPi}
is positive definite under the assumption
  \begin{equation}\label{condpi}
  \frac{\tau}{Z_{rot}\tauc}\frac{c_v^{rot}}{c_v^{tr}}
  +  \frac{\tau}{Z_{vib}\tauc}\frac{c_v^{vib}(T_{tr})}{c_v^{tr}}
  <\frac{3}{\Prandtl}(\Prandtl- \frac23). 
\end{equation}
\end{proposition}

\begin{proof}
First, note that $\Pi$ and $\Theta$ have the same eigenvectors, and
hence relation~\eqref{eq-defPi} written in this eigenvector basis
reads
\begin{equation*}
  \lambda_i(\Pi) = R \Ttrrel + (1-\frac{1}{\Prandtl})(\lambda_i(\Theta)-R\Ttr),
\end{equation*}
where $\lambda_i(\Pi)$ and $\lambda_i(\Theta)$ are the eigenvalues of
$\Pi$ and $\Theta$ for $i=1$, $2$, $3$. By~\eqref{eq-T_tr_rot_vib}
and~\eqref{eq-Theta_q}, we have
$R\Ttr = \frac13 (\lambda_1(\Theta)+
\lambda_2(\Theta)+\lambda_3(\Theta))$, and since the
$\lambda_i(\Theta)$ are positive (note that~\eqref{eq-Theta_q} implies
$\Theta$ is positive definite), we get
$\lambda_i(\Theta) \leq 3 R\Ttr$. Finally, the assumption
$\Prandtl \leq 1$ implies
\begin{equation}\label{eq-minlambdaiPi} 
   \lambda_i(\Pi) \geq  R \Ttrrel + (1-\frac{1}{\Prandtl})2R\Ttr.
 \end{equation}
Consequently, a sufficient condition for $\Pi$ to be positive definite is that
 the right-hand side of~\eqref{eq-minlambdaiPi} is positive.

 Now, we inject the expression of $\Ttrrel$~\eqref{eq-cond1}
 into~\eqref{eq-minlambdaiPi}, and we find that the right-hand side
 of~\eqref{eq-minlambdaiPi} is positive if
\begin{equation}  \label{eq-cond2}
  \frac{3}{\Prandtl} (\Prandtl - \frac23) \Ttr
  -\frac{\tau}{\Zrot\tauc}\frac{c_v^{rot}}{c_v^{tr}}(\Ttr - \Trot)
  - \frac{\tau}{\Zvib\tauc}\frac{c_v^{vib}(T_1)}{c_v^{tr}}(\Ttr - \Tvib) \geq 0.
\end{equation}
The same analysis as for the proof of proposition~\ref{prop:pos_ener}
gives the final condition~\eqref{condpi}.

\end{proof}

\begin{remark}
  Condition~\eqref{condpi} is clearly not optimal, since the
  directional temperatures $\lambda_i(\Theta)/R$ are generally close
  to $\Ttr$ and the non zero values of $\Trot$ and $\Tvib$ help in
  getting~\eqref{eq-cond2} (see~\cite{Mathiaud2016} for an optimal
  condition obtained in the monoatomic case).
\end{remark}

\subsection{Conservation properties}
\label{subsec:conservation}
\begin{proposition}
The collision operator~\eqref{eq: operator} of the ES-BGK model satisfies the conservation
of mass, momentum, and energy:
\begin{equation}\label{conserv}
  \la (1,v,\frac12 |v-u|^2+ \varepsilon + iRT_0)
\frac{1}{\tau}({\cal G}[f]-f) \ra_{v,\varepsilon,i}  = 0.
\end{equation}
\end{proposition}
\begin{proof}
This is a simple consequence of the definition of the relaxation variables $T_{\alpha}^{rel}$ and $\Pi$
(see~(\ref{eq-Trotrel}--\ref{eq-defPi})), and of the integral relations~(\ref{eq-Gaussalpha}--\ref{eq-intGvib}).
\end{proof}

\subsection{\label{sec: entropy}Entropy}

The use of a single rotational energy with $\delta$ degrees of freedom
requires to define the Boltzmann entropy functional as
\begin{equation*}
   \ent(f) = \langle f\log (f/\varepsilon^{\frac{\delta}{2}-1}) -f\rangle_{v,\varepsilon,i}.
 \end{equation*}
For any macroscopic values $(\rho,u,\Theta,\Trot,\Tvib)$, we define the following set of distribution functions that realizes these values, namely
 \begin{equation}\label{eq-setX} 
      \begin{split}
     {\cal X}_{\rho,u,\Theta,\Trot,\Tvib} =
     \lbrace
   & \phi\geq 0, \quad 
    \la (1+|v|^2+\varepsilon+i+|\log(\phi/\varepsilon^{\frac{\delta}{2}-1})|)\phi\ra_{v,\varepsilon,i}<+\infty, \\
&   \la
     (1,v,(v-u) \otimes (v-u),\varepsilon, iRT_0) \phi\ra_{v,\varepsilon,i}
 =\left(\rho,\rho u,\rho\Theta,\rho \erot(\Trot),\rho
   \evib(\Tvib) \right)\rbrace.
\end{split}
     \end{equation}
Now we state the H-theorem for our model.

\begin{proposition} \label{prop:theoH}
  We assume $\frac23<\Pr\leq 1$ and conditions~\eqref{eq-erelpositive} and~\eqref{condpi} are satisfied. Our ES-BGK model~\eqref{eq: cinetique}-\eqref{eq: operator} satisfies
\begin{equation}\label{eq-Hineg} 
  \partial_t \ent(f)+ \nabla \cdot  \la v ( f\log(f/\varepsilon^{\frac{\delta}2-1} )-f)\ra_{v,\varepsilon,i}=\la
  \frac{1}{\tau}({\mathbf{\mathcal{G}}}[f]-f)\log(f/\varepsilon^{\frac{\delta}2-1} )\ra_{v,\varepsilon,i} \leq 0,
\end{equation}
under the additional condition
\begin{equation}  \label{eq-condZ}
\frac{\tau}{\Zrot\tauc} + \frac{\tau}{\Zvib\tauc} \leq
\frac{3}{5}. 
\end{equation}
Moreover, the right-hand side of~\eqref{eq-Hineg} is zero
if, and only if $f = {\mathbf{\mathcal{M}}}[f]$.
\end{proposition}

\begin{proof}
  We remind elements of proof already proved in~\cite{DMM_2021} that apply here too:
  \begin{enumerate}
  \item The Gaussian distribution
   ${\mathbf{\mathcal{G}}}[f]$ defined by~\eqref{eq: operator} is the
   unique minimizer of the entropy functional
   $\ent(f)$ on the set ${\cal X}_{\rho,u,\Pi,\Trotrel,\Tvibrel}$, defined according to~\eqref{eq-setX}.

   \item By convexity of $\ent$, the right-hand side of~\eqref{eq-Hineg} is non positive under the sufficient condition
     \begin{equation}  \label{eq-HGf}
\ent({\cal G}[f]) \leq \ent(f).
\end{equation}
This condition is not obvious, since $f$ is not in ${\cal X}_{\rho,u,\Pi,\Trotrel,\Tvibrel}$.

\item For any macroscopic quantities $(\rho,u,\Theta,\Trot,\Tvib)$, we denote by
  $S(\rho,u,\Theta,\Trot,\Tvib)$ the minimum value of $\ent$ on ${\cal X}_{\rho,u,\Theta,\Trot,\Tvib}$, and we have
  \begin{equation}\label{eq-defSent} 
    S(\rho,u,\Theta,\Trot,\Tvib) =   \rho \log \rho +  C \rho
- \frac{\rho}{R} {\cal S}(\frac32 (\det \Theta)^{\frac13},\erot(\Trot),\evib(\Tvib))
\end{equation}
where ${\cal S}$ is the entropy at constant density defined in
section~\ref{sec:thermo}, and $C$ is a constant that depends on $\delta$ and $R$ only.

\item By point 1, we have $\ent({\cal G}[f]) = S(\rho,u,\Pi,\Trotrel,\Tvibrel)$.

\item Since $f$ is in ${\cal X}_{\rho,u,\Theta,\Trot,\Tvib}$, then we have $S(\rho,u,\Pi,\Trotrel,\Tvibrel) \leq \ent(f)$.
  
\item A sufficient condition for~\eqref{eq-HGf} is therefore
\begin{equation}  \label{eq-conds1}
  S(\rho,u,\Pi,\Trotrel,\Tvibrel) \leq S(\rho,u,\Theta,\Trot,\Tvib).
  \end{equation}

\item We have
  \begin{equation}\label{ineqthetapi} 
\frac{\det \Theta}{\det\Pi}\leq \left(\frac{\etr(\Ttr)}{\etr(\Ttrrel)}\right)^3.
\end{equation}
The proof of this inequality is slightly different from that
shown in~\cite{DMM_2021} and is given in appendix~\ref{subsec:ineqthetapi}.
  \item With points 3, 6, and 7, a sufficient condition for~\eqref{eq-conds1} is
    \begin{equation}  \label{eq-conds2}
      \begin{split}
  & {\cal S} (\etr(\Ttr))  , \srot(\erot(\Trot)) ,\svib(\evib(\Tvib)) \\
  & \leq
    {\cal S} (\etr(\Ttrrel))  , \srot(\erot(\Trotrel)) ,\svib(\evib(\Tvibrel)),
  \end{split}
\end{equation}

\end{enumerate}
The proof of this last inequality is the only part which is different from~\cite{DMM_2021}, and a bit more involved. Our proof is divided into 5 steps.

\paragraph{Step 1: parametrization of $\cal S$.}
We consider $(\etr(\Ttrrel),\erot(\Trotrel),\evib(\Tvibrel)$ as
(affine) functions of parameters $\Zrot$ and $\Zvib$, and we set
\begin{equation}  \label{eq-def_h}
h(\thetarot,\thetavib) = {\cal S}(\etr(\Ttrrel),\erot(\Trotrel),\evib(\Tvibrel) ),
\end{equation}
where $\thetarot=\frac{\tau}{\Zrot\tauc }$ and
$\thetavib=\frac{\tau}{\Zvib\tauc }$. With these new parameters, we have
\begin{align}
  & \etr(\Ttrrel) = \etr(\Ttr)
    + \thetarot(\erot(\Trot) -\erot(\Ttr) )
    + \thetavib(\evib(\Tvib) -\evib(\Ttr) )  \label{eq-def_Ttrrel} \\
   &  \erot(\Trotrel) = \thetarot \erot(\Ttr) + (1- \thetarot) \erot(\Trot),
  \label{eq-def_Trotrel} \\
 &   \evib(\Tvibrel) = \thetavib \evib(\Ttr) + (1- \thetavib) \evib(\Tvib)
  \label{eq-def_Tvibrel}.
  \end{align}

Now it it clear that for $(\thetarot,\thetavib)=(0,0)$, the relaxation
energies reduce to the initial energies, that is to say
\begin{equation*}
  (\etr(\Ttrrel),\erot(\Trotrel),\evib(\Tvibrel)
  )|_{(\thetarot,\thetavib)=(0,0)} = (\etr(\Ttr),\erot(\Trot),\evib(\Tvib)).
\end{equation*}
Consequently, our entropy inequality~(\ref{eq-conds2}) reads
\begin{equation}  \label{eq-ineqh}
h(0,0) \leq h(\thetarot,\thetavib).
  \end{equation}

  While the domain of definition of $h$ is given by positiveness
  condition~\eqref{eq-erelpositive}, here we need to reduce it to
  condition~\eqref{eq-condZ} given in the proposition. With our new parameters, it reads
\begin{equation}\label{eq-domaine_reduit_h} 
\thetarot + \thetavib\leq \frac35.
  \end{equation}
In fact, numerical tests suggest~(\ref{eq-ineqh}) can be false if this condition is not fulfilled.
  
Finally, note that $h$ is concave, as composed of an affine function and the
concave function $\cal S$.  

  \paragraph{Step 2: relaxation temperatures as convex combinations}

  Here we use the same argument as used in the proof of proposition~\ref{prop:pos_ener}: we linearize~\eqref{eq-def_Ttrrel}--\eqref{eq-def_Tvibrel} by using the mean value theorem, and we get
  \begin{equation} \label{eq-Talphaconvex}
    \begin{split}
      \Ttrrel & = (1 - \thetarot \frac{\cvrot}{\cvtr}
                    - \thetavib \frac{\cvvib(T_1)}{\cvtr}) \Ttr
               + \thetarot \frac{\cvrot}{\cvtr} \Trot
                + \thetavib \frac{\cvvib(T_1)}{\cvtr} \Tvib, \\
       \Trotrel & = \thetarot  \Ttr + (1-\thetarot ) \Trot, \\
       \Tvibrel & = \thetavib \frac{\cvvib(T_1)}{\cvvib(T_2)} \Ttr
               + (1 - \thetavib \frac{\cvvib(T_1)}{\cvvib(T_2)}) \Tvib,  
    \end{split}
  \end{equation}
    where $T_1$ and $T_2$ are some temperatures between $\Ttr$ and $\Tvib$, and $\Tvibrel$ and $\Tvib$, respectively, defined by
    \begin{equation}\label{eq-defT1T2} 
      \cvvib(T_1) = \frac{\evib(\Ttr) - \evib(\Tvib)}{\Ttr - \Tvib}, \qquad \text{and} \qquad 
      \cvvib(T_2) = \frac{\evib(\Tvib) - \evib(\Tvibrel)}{\Tvib - \Tvibrel}.
      \end{equation}

\paragraph{Step 3: minimization of $h$.}

In the plane $(\thetarot,\thetavib)$, condition~\eqref{eq-domaine_reduit_h} defines a
triangle $\cal T$ of vertices $(0,0)$, $(0,\frac{3}{5})$, $(\frac{3}{5},0)$.
Since $h$ is concave, its minimum on $\cal T$ is reached at one vertex
of $\cal T$. Therefore
\begin{equation*}
  h(\thetarot,\thetavib) \geq \min(h(0,0),h(0,\frac{3}{5}),h(\frac{3}{5},0)).
\end{equation*}
for every $(\thetarot,\thetavib)$ in $\cal T$.

Then a sufficient condition for~\eqref{eq-ineqh} is that the minimum is reached at $(0,0)$, that is to say
\begin{equation}  \label{eq-condvertices}
  h(0,0) \leq h(0,\frac{3}{5}) \qquad \text{ and } \qquad h(0,0) \leq h(\frac{3}{5},0).
\end{equation}

Now, we prove the first inequality of~(\ref{eq-condvertices}). In fact, it is simpler to prove a stronger property, namely
\begin{equation} \label{eq-condvertices2}
  h(0,0) \leq h(0,\thetavib)
\end{equation}
for every $\thetavib\leq \frac35$. We start by using that $h$ is
concave but also that it is differentiable to get
\begin{equation*}
 h(0,\thetavib) \geq h(0,0) + \thetavib \partial_{\thetavib} h(0,\thetavib),
\end{equation*}
and now a sufficient condition to get~(\ref{eq-condvertices2}) is
$\partial_{\thetavib} h(0,\thetavib)\geq 0$. But the chain rule
gives
\begin{equation*}
  \begin{split}
  \partial_{\thetavib} h(\thetarot,\thetavib)
&   = \nabla {\cal S}(\etr(\Ttrrel),\erot(\Trotrel),\evib(\Tvibrel) )
  \cdot\frac{\partial}{\partial \thetavib}
  \begin{pmatrix}
     \etr(\Ttrrel)\\
     \erot(\Trotrel)\\
     \evib(\Tvibrel
  \end{pmatrix} \\
  & = \Bigl(\frac{1}{\Ttrrel} - \frac{1}{\Tvibrel} \Bigr) (\evib(\Tvib) - \evib(\Ttr)). 
  \end{split}
\end{equation*}
Now for $\thetarot=0$ this reads
\begin{equation}\label{eq-gradh} 
   \partial_{\thetavib} h(0,\thetavib) = (\frac{1}{{T}_{tr}^{rel,0}} - \frac{1}{T_{vib}^{rel}} ) (\evib(\Tvib) - \evib(\Ttr)),
 \end{equation}
 where the exponent $0$ indicates that ${T}_{tr}^{rel,0}$ is defined
 by~\eqref{eq-def_Ttrrel} with $\thetarot=0$.

 This quantity can be proved to be non negative if we are able to show
 that $T_{vib}^{rel} - {T}_{tr}^{rel,0} = \alpha
 (\Tvib - \Ttr)$ with $\alpha \geq 0$. Indeed, if the second bracket
 of~(\ref{eq-gradh}) is positive, then $\Tvib - \Ttr \geq 0$ since
 $\evib$ is an increasing function, and hence $T_{vib}^{rel}
 - {T}_{tr}^{rel,0}\geq 0$ too, and the first bracket
 of~(\ref{eq-gradh}) is positive as well, which gives the sign of
 $\partial_{\thetavib} h(0,\thetavib)$. The proof is the same in the
 opposite case.

 The relation $\Tvibrel - {T}_{tr}^{rel,0}  = \alpha
 (\Tvib - \Ttr)$ is obtained with~\eqref{eq-Talphaconvex}. With
 $\thetarot= 0$, these relations give
\begin{equation*}
  \Tvibrel - {T}_{tr}^{rel,0}
  = \left(  1 - \thetavib  \frac{\cvvib(T_1)}{\cvvib(T_2)} 
              - \thetavib \frac{\cvvib(T_1)}{\cvtr}\right)  (\Tvib - \Ttr). 
\end{equation*}
Our coefficient $\alpha$ is clearly non negative under the condition
\begin{equation}\label{eq-condcvvib} 
  \thetavib  \frac{\cvvib(T_1)}{\cvvib(T_2)} +  \thetavib \frac{\cvvib(T_1)}{\cvtr}  \leq 1.
\end{equation}
Now, this condition is analyzed with two different cases.

\paragraph{First case:  $\Ttr\leq  \Tvib$.} Since $\Tvibrel$ is a convex combination of $\Ttr$ and $\Tvib$, we have $\Ttr\leq \Tvibrel\leq \Tvib$. Then the intermediate temperatures $T_1$ and $T_2$ are in intervals $[\Ttr,\Tvib] $ and $[\Tvibrel,\Tvib]$, respectively. Now, the convexity of $\evib$ implies $\cvvib(T_2)\geq \cvvib(T_1)$ (see appendix~\ref{app:evibconvex}). Consequently, the first term of the left-hand side of~\eqref{eq-condcvvib} satisfies
\begin{equation*}
  \thetavib \frac{\cvvib(T_1)}{\cvvib(T_2)} \leq   \thetavib.
\end{equation*}
Moreover, the second term satisfies $\thetavib
\frac{\cvvib(T_1)}{\cvtr} \leq \frac23 \thetavib$ (since $\cvvib$ is bounded by $R$, see section~\ref{subsec:prop_energies}). Finally, the left-hand side of~\eqref{eq-condcvvib} satisfies
\begin{equation*}
  \thetavib  \frac{\cvvib(T_1)}{\cvvib(T_2)} +  \thetavib \frac{\cvvib(T_1)}{\cvtr}  \leq \frac53 \thetavib
\end{equation*}
which is indeed lower than 1, since $\thetavib\leq 3/5$. Therefore~\eqref{eq-condcvvib} is satisfied.

\paragraph{Second case: $\Ttr\geq  \Tvib$.} Now we have $\Tvib\leq
\Tvibrel\leq \Ttr$. This case is more delicate, since in the first
term of~\eqref{eq-condcvvib}$, \frac{\cvvib(T_1)}{\cvvib(T_2)}$ now is
greater than 1. Thus we must work on the product
$\thetavib\frac{\cvvib(T_1)}{\cvvib(T_2)}$. By
using~\eqref{eq-defT1T2}, we have
\begin{equation}\label{eq-diffrel1}
  \begin{split}
  \thetavib  \frac{\cvvib(T_1)}{\cvvib(T_2)}
  &   = \thetavib  \left(\frac{\evib(\Ttr) - \evib(\Tvib)}{\evib(\Tvib) - \evib(\Tvibrel)}\right) \left(\frac{\Ttr - \Tvib}{\Tvib - \Tvibrel}\right) \\
  & = \frac{\Tvibrel-\Tvib}{\Ttr-\Tvib}
  =   1 -   \frac{\Ttr - \Tvibrel}{\Ttr-\Tvib},
\end{split}
\end{equation}
where we have used~\eqref{eq-def_Tvibrel} to simplify the energy ratio.

  Now, note that
$\evib(\Ttr)-\evib(\Tvibrel)\leq \cvvib(\Ttr)(\Ttr - \Tvibrel)$,
since $\cvvib$ is bounded by $\cvvib(\Ttr)$ in $[\Tvib,\Ttr]$, and hence
\begin{equation*}
  \begin{split}
 \Ttr - \Tvibrel & \geq \frac{ \evib(\Ttr)-\evib(\Tvibrel)}{\cvvib(\Ttr)} \\
&  = \frac{(1-\thetavib)(\evib(\Ttr) - \evib(\Tvib))}{\cvvib(\Ttr)} 
= (1-\thetavib)\frac{\cvvib(T_1)}{\cvvib(\Ttr)}(\Ttr-\Tvib),
\end{split}
\end{equation*}
 from~\eqref{eq-def_Tvibrel} and~\eqref{eq-defT1T2}. Consequently, we go back to~\eqref{eq-diffrel1} and we get
 \begin{equation*}
   \thetavib  \frac{\cvvib(T_1)}{\cvvib(T_2)}
   \leq 1 - (1-\thetavib)\frac{\cvvib(T_1)}{\cvvib(\Ttr)},
 \end{equation*}
 which is now clearly lower than 1.
Therefore, a sufficient condition for~\eqref{eq-condcvvib} is
 \begin{equation*}
   1 - (1-\thetavib)\frac{\cvvib(T_1)}{\cvvib(\Ttr)} + \thetavib \frac{\cvvib(T_1)}{\cvtr}  \leq 1,
   \end{equation*}
   which is equivalent to
   \begin{equation*}
     \thetavib \leq \frac{1}{1 + \frac{\cvvib(\Ttr)}{\cvtr}}.
   \end{equation*}
   Now, since the ratio $\frac{\cvvib(\Ttr)}{\cvtr}$ is lower than $2/3$, this last inequality is satisfied if $ \thetavib \leq 1 / (1 + 2/3) =3/5$, which is what we wanted to prove, and hence~\eqref{eq-condcvvib} is now proved for every cases. This proves~\eqref{eq-condvertices2} for every $\thetavib\leq \frac35$ and hence the first inequality of~\eqref{eq-condvertices} is proved.

 The second inequality of~\eqref{eq-condvertices} is proved in a similar way, but much more easily, since $\erot$ is linear. Indeed, the ratio $\cvrot(T_1)/\cvrot(T_2)$ in the equivalent of~\eqref{eq-condcvvib} is equal to 1, and the inequality is obviously satisfied for every $\thetarot\leq 3/5$.

 This long analysis proves~\eqref{eq-condvertices}, and
 hence~\eqref{eq-ineqh} and in turn~\eqref{eq-conds2}. The proof of
 the proposition is now almost complete: the equilibrium part is
 proved like in~\cite{DMM_2021} and is left to the reader.

\end{proof}

\subsection{Discussion on the conditions for positiveness of relaxation energies, positive definiteness of $\Pi$, and H-theorem}

\paragraph{Hierarchy of conditions.}
Propositions~\ref{prop:pos_ener},~\ref{pipositive},
and~\ref{prop:theoH} hold for different conditions that are in fact
not completely independent.

For instance condition~\eqref{condpi} for positive
definiteness of $\Pi$ implies the third constraint of
condition~\eqref{eq-erelpositive} for the positiveness of relaxation energies: indeed, the right-hand side
of~\eqref{condpi} is lower than 1 for $\Prandtl \leq 1$.

Moreover, in proposition~\ref{prop:theoH}, condition~\eqref{eq-condZ}
clearly implies the first two constraints
of~\eqref{eq-erelpositive}. However, it does not always
implies the third constraint of~\eqref{eq-erelpositive}, since it is
temperature dependent.

This means that the number of conditions could be reduced in our
propositions.
Nevertheless, we find that the current redundancy is clearer, since
there is a clear hierarchy: for the H-theorem to hold, we should
first assume that the relaxation energies are positive,
and then that $\Pi$ is positive definite.

\paragraph{Physical validity.}
Now, we discuss the physical validity of these conditions. As an
example, we consider a flow of nitrogen, for which the characteristic
vibrational temperature is $T_0=3.371$K, and the molecular VSS
parameters are $\omega=0.74$ and $\alpha=1.36$. With the Eucken
formula
$\Prandtl = 2(5+\delta + \delta_v)/(15 + 2 (\delta + \delta_v))$ and
definitions~\eqref{eq-icvib} and~\eqref{eq-taucvss}, we can compute
all the terms of conditions~\eqref{eq-erelpositive},~\eqref{condpi},
and~\eqref{eq-condZ}, for any arbitrary temperature, and hence we can
check for what range of temperature these conditions are
satisfied. For $\Zvib$, we use the Millikan-White formula as given
in~\cite{bird, Millikan1963}. For $\Zrot$, its usual value in
aerodynamics is $\Zrot=5$, but we also use its value as given by the
Parker formula~\cite{bird,parker1959}.  Our observations are as
follows.

For $\Zrot=5$, all the conditions are satisfied up to a temperature of
$40.000$ K. For larger temperatures, the constraint
$\tau/\Zvib\tauc<1$ of~\eqref{eq-erelpositive} fails, and the vibrational
energy becomes negative. This upper bound is clearly sufficient here, since the model is not
designed for so large temperatures, for which other physical
phenomenon have to be taken into account (dissociation for instance). 
In addition, the model of Millikan and White~\cite{Millikan1963} is an empirical model which in the original paper itself is only defined in a temperature range between $280\,\mathrm{K} <T<8000\,\mathrm{K}$, so that the physical suitability at $T=40.000\,\mathrm{K}$ may be doubted. In general, the physical suitability of the model for very high temperatures is doubtful, since $\Zvib$ then approaches 0. However, a $\Zvib<1$ would be problematic from a purely physical point of view, since the relaxation time would then be smaller than the collision time itself.

For $\Zrot$ as given by Parker formula, note that $\Zvib$ and $\Zrot$
behave very differently, since $\Zrot$ increases with the temperature, while
$\Zvib$ decreases very fast, and is infinitely large for small
temperatures. Then we observe that all the conditions are satisfied for
temperatures between $60$ and $42.000$ K. Again, the upper bound is
clearly sufficient. The lower bound is due to the constraint
$\tau/\Zrot\tauc<1$ of~\eqref{eq-erelpositive}: for lower
temperatures, this constraint is not satisfied, and the rotational
energy becomes negative (the other conditions fail for small
temperatures a bit smaller, between $20$ and $32$, which is less
restrictive). Here the same problem arises as already described for the vibration, since $\Zrot$ goes towards 0 for decreasing temperatures. Again, $\Zrot<1$ is difficult from a purely physical point of view. The model is therefore not suitable for such low temperatures. There is another problem: the characteristic rotational temperature of N$_2$ is $T_{0,rot}=2.88\,\mathrm{K}$. For hydrogen H$_2$, for example, this is already $T_{0,rot}=87.6\,\mathrm{K}$.  At such low temperatures, one can no longer necessarily assume that the rotational degree of freedom is fully excited, which means that the number of degrees of freedom of the rotation and thus $\cvrot$ also become temperature-dependent for very low temperatures, comparable with the vibration in the considered temperature range. In the model proposed here, however, this effect was not taken into account, as these temperatures are lower than the smallest temperatures generally met in aerodynamics. Therefore, this effect is typically also neglected in DSMC codes and the rotational temperature is assumed to be continuous.
In general, very little information can be found in the literature about the relaxation time of rotation at very low temperatures. However, in Riabov~\cite{riabov2011} one can find a discussion about the discrepancy between the classical consideration of Parker's model and the technique of Lebed and Riabov~\cite{lebed1979quantum} for the relaxation times of rotation for $T<100\,\mathrm{K}$. It becomes clear that the Parker model can no longer be used for these low temperatures.







\subsection{Comparison with the ES-BGK model of Andri\`es et al.~\cite{ALPP}}
\label{subsec:compAndries}
If the vibration modes are neglected, our model reduces to the
following translation-rotation ES-BGK model:
\begin{equation}  \label{eq-ESBGKrotP1}
\partial_t f + v \cdot \nabla_x f = \frac{1}{\tau} (\G[f] -f),
\end{equation}
where now $f$ does not depend on $i$, while the Gaussian is $\G[f] =
\Gtr[f] \Grot[f]$, with $\Pi$ and $\Trotrel$ are still defined
by~\eqref{eq-defPi} and~\eqref{eq-Trotrel}, and $\Ttrrel$ is now
defined by
\begin{equation}  \label{eq-ESBGKrotP2}
\etr(\Ttrrel) = \etr(\Ttr)
  -\frac{\tau}{\Zrot \tauc} (\erot(\Ttr) - \erot(\Trot)),
\end{equation}
while $\Tvibrel$ is not used anymore. The macroscopic quantities are
defined as in~\eqref{eq-mtsf}--\eqref{eq-Theta_q} without the series
in $i$. Here, the model is not restricted to diatomic gases anymore,
and $\delta$ can take any integer values greater than or equal to 2.

For such polyatomic gases, the first ES-BGK model was proposed by
Andri\`es et al~\cite{ALPP}, and is often used in the literature (see
for instance~\cite{KKA_2019}). This model reads as above with
relaxation tensor
\begin{equation}  \label{eq-PiAndries}
\Pi = (1-\theta) ((1-\nu) R \Ttr I + \nu \Theta) + \theta R \Teq,
\end{equation}
and the relaxation rotational temperature is
\begin{equation}  \label{eq-TrotrelAndries}
\Trotrel = \theta \Teq + (1-\theta) \Trot,
\end{equation}
where the equilibrium temperature is
\begin{equation}  \label{eq-TeqAndries}
\Teq = \frac{3\Ttr + \delta \Trot}{3 + \delta}.
\end{equation}
Note that in~\cite{ALPP}, $\Trot$ is denoted by $T_{int}$, $\Trotrel$ by
$T_{rot}^{rel}$, and $\Pi$ by ${\cal T}$. Moreover, the variable
$I=\varepsilon^{\delta/2}$ is used instead of $\varepsilon$, which
does not change our analyzis and conclusions below. Finally, the
parameters $\theta$ and $\nu$ are defined by
\begin{equation}  \label{eq-nuthetaAndries}
\theta = \frac{1}{\Zrot}, \qquad \text{ and } \qquad (1-\theta)\nu =
\frac{1}{\Prandtl} -1.
\end{equation}

First, we show that our model can be written under the same form as
the Andri\`es et al. model, with modified coefficients.
\begin{proposition}
  The relaxation tensor $\Pi$ and rotational temperature of
  model~\eqref{eq-ESBGKrotP1}-\eqref{eq-ESBGKrotP2} can be written under
  form~\eqref{eq-PiAndries} and~\eqref{eq-TrotrelAndries} with
  modified coefficients $\tilde{\theta}$ and $\tilde{\nu}$ defined by
  \begin{equation*}
    \tilde{\theta} = \frac{3+\delta}{3} \frac{\tau}{\tauc}\theta\qquad \text{ and } \qquad 
    (1-\tilde{\theta})\tilde{\nu} = \frac{1}{\Prandtl} -1.
  \end{equation*}
\end{proposition}
This proposition is readily proved with a direct calculation in which
$\Ttr$ is written as a function of $\Teq$ and $\Trot$ by
using~\eqref{eq-TeqAndries}. This is left to the reader.

It is interesting to compare coefficients $\theta$ and $\tilde{\theta}$ of
both models. For instance, for a diatomic gas ($\delta=2$) with a
ratio $\tau/\tauc\approx 1.7$ in case of the HS
collision model (see section~\ref{subsec:chartime} with the value
$\Prandtl \approx 0.74$ for a diatomic gas without vibration modes),
we find $\tilde{\theta}$ is approximately $3\theta$. This shows that
these two models have very different coefficients.

Another way to compare these models is to look at energy
relaxations. The following proposition compares relaxation times for
both models.
\begin{proposition}
  For both ES-BGK models, in the space
  homogeneous case, the rotational temperature relaxes according to
  \begin{equation}\label{eq-Taurot} 
    \frac{d}{dt}\Trot = \frac{1}{\taurot} (\Ttr-\Trot).
  \end{equation}
  where the relaxation time is
  \begin{equation*}
    \begin{split}
& \taurot = \Zrot \tauc \qquad \text{ for
  our model~\eqref{eq-ESBGKrotP1}--\eqref{eq-ESBGKrotP2}, and } \\
& \taurot = \Zrot \tau (3+\delta)/3 \qquad \text{ for Andri\`es et al. model.}    
\end{split}
\end{equation*}
\end{proposition}
This proposition is proved by a direct integration of the homogeneous
kinetic equation times $\varepsilon$, and then by using the definition
of $\Trotrel$. This result shows that both models give different
exchange rates of energy between rotational and translational
modes. This can be seen more clearly with the previous example of a
diatomic gas, since we find $\taurot|_{Andries}\approx 3 \taurot$, which
means that the rotational energy of Andri\`es et al. model relaxes three
times as fast as with our model. If the correct relaxation rate is
$\taurot$ (as it is used in some DSMC codes,
see~\cite{BS_2017,pfeiffer2018,Haas1994,parker1959}), then the energy exchange rate as
given by Andri\`es et al. model is much too large.

\section{Hydrodynamic asymptotics \label{sec: chap}}
To obtain the conservation laws, we multiply~(\ref{eq:
  cinetique}) by the vector $1$, $v$, and
$\frac{1}{2}|v|^2+\varepsilon+iRT_0$ and we integrate to get:
\begin{equation}
\label{eq: conservation}
\begin{aligned}
&\partial_t \rho+\nabla \cdot (\rho u)=0,\\
&\partial_t (\rho u)+\nabla \cdot (\rho u \otimes u)+ \nabla\cdot\PSigma(f)=0,\\
&\partial_t \mathcal{E}+\nabla \cdot (\mathcal{E}u)+ \nabla \cdot
\left(\PSigma(f)u\right) + \nabla \cdot q(f)=0,
\end{aligned}
\end{equation}
where $\mathcal{E} =   \langle(
\frac{1}{2}|v|^2+\varepsilon+iRT_0)f\rangle_{v,\varepsilon,i}=\frac{1}{2}\rho|u|^2+\rho E(f) $ is the
total energy density, while the pressure tensor $\PSigma(f)$ and the
heat flux $q(f)$ have been defined by~\eqref{eq-Theta_q}.
If we have some characteristic values of length, time,
velocity, density, and temperature, our ES-BGK model~\eqref{eq:
  cinetique}--\eqref{eq: operator} can be non-dimensionalized. This
equation reads
\begin{equation}
\label{eq: nd_kinetic}
\partial_t f+v\cdot \nabla f=\frac{1}{\Kn \, \tau }({\mathbf{\mathcal{G}}}[f]-f),
\end{equation}
where $\Kn$ is the Knudsen number which is the ratio between the mean
free path and a macroscopic length scale. For simplicity, we use the
same notations for the non-dimensional and dimensional variables. Note that we assume here that the three
  relaxation times have the same asymptotic order of magnitude with
  respect to $\Kn$ (even if their values can be very different).

The Chapman-Enskog analysis consists in approximating the pressure
tensor and the heat flux at zero and first order with
respect to the Knudsen number, leading to compressible Euler
equations and compressible Navier-Stokes equations, respectively.

\subsection{Euler asymptotics}

We get the following
proposition, that can be proved as in~\cite{DMM_2021}.
\begin{proposition}\label{prop:Euler}
The moments of $f$, solution of the ES-BGK
  model~\eqref{eq: cinetique}, satisfy the compressible Euler
  equations up to $O(\Kn)$:
  \begin{equation}\label{eq-euler} 
\begin{aligned}
&\partial_t \rho+\nabla \cdot(\rho u)=0,\\
&\partial_t (\rho u)+\nabla \cdot(\rho u \otimes u)+\nabla p=O(\Kn),\\
&\partial_t \mathcal{E}+\nabla \cdot\left((\mathcal{E}+p)u\right)=O(\Kn),
\end{aligned}
\end{equation}
where $p=\rho R T_{eq}$ is the pressure at equilibrium.
The non-conservative form of these equations is
\begin{equation}\label{eq-euler_nc} 
\begin{aligned}
&\partial_t \rho+u\cdot \nabla\rho+\rho\nabla \cdot u=0,\\
&\partial_t u+(u\cdot \nabla)u+\frac{1}{\rho}\nabla p=O(\Kn),\\
&\partial_t T_{eq}+u\cdot \nabla T_{eq}+(\gamma -1)T_{eq}  \nabla \cdot u=O(\Kn),
\end{aligned}
\end{equation}
where $\gamma = c_p(\Teq)/c_v(\Teq)$ is the ratio of specific heats,
with $c_p(\Teq) = c_v(\Teq) + R$  and $c_v(\Teq) = \frac{de(\Teq)}{dT}
= \cvtr + \cvrot + \cvvib(\Teq)$ is the specific heat at constant volume.  
\end{proposition}

\subsection{\label{sec: NS}Compressible Navier-Stokes asymptotics}

Our main result is the following.
\begin{proposition} \label{prop:CNS} The moments of $f$, solution of the ES-BGK
  model~\eqref{eq: cinetique}, satisfy the compressible Navier-Stokes
  equations up to $O(\Kn^2)$:
\begin{align*}
&\partial_t \rho + \nabla \cdot (\rho u)=O(\Kn^2), \\
  &\partial_t (\rho u)+\nabla \cdot (\rho u\otimes u)
    +\nabla p= \nabla\cdot \sigma +O(\Kn^2),\\
  &\partial_t \mathcal{E}+\nabla \cdot (\mathcal{E}+p)u
    =-\nabla \cdot q + \nabla \cdot(\sigma u)+O(\Kn^2),
\end{align*}
where, in dimensional form, the viscous stress tensor and the heat flux are given by
\begin{equation*}
\sigma= \mu\left(\nabla u+(\nabla u)^T-\frac23 \nabla \cdot u
    I\right)+ \zeta \nabla \cdot u I,\qquad  q=-\kappa\nabla T,
\end{equation*}
and the viscosity, heat transfer and  volume viscosity coefficients are
\begin{equation*}
  \mu={\tau p\Prandtl},\qquad \kappa=\frac{\mu c_p(\Teq)}{\Prandtl},
  \qquad\zeta =
  pR \left(\taurot\frac{\cvrot}{\cv(T_{eq})^2} + \tauvib\frac{\cvvib(\Teq)}{\cv(T_{eq})^2}\right),
\end{equation*}
with $\taurot = \Zrot\tauc$ and $\tauvib = \Zvib\tauc$, while $\Zrot$,
$\Zvib$, $\tauc$, and $\tau$ are defined at $\Teq$.
\end{proposition}
For the proof of proposition~\ref{prop:CNS}, most of the calculations are very similar
to that given in~\cite{DMM_2021}. The only difference is the first
order expansion of the temperatures and of the tensor $\Pi$ required
to compute the Chapman-Enskog expansion. The
corresponding procedure and results are given in appendix~\ref{sec:hydrodynamic-limits}.

\begin{remark}
  The volume viscosity $\zeta$ is the same as that found by Bruno and
  Giovangigli in~\cite{domenico} for
a one temperature Navier-Stokes asymptotics derived from a
Boltzmann equation for a diatomic gas with two internal modes. Indeed, in~\cite{domenico},
when we assume that the vibrational and rotational modes are
independent and that $\tau_{vib}$ and $\tau_{rot}$ are of the same
order as $\tauc$, then equation (A2) of~\cite{domenico} with
$rap=vib$, $sl=rot$, $K^{vib,rot}=0$, and $\tau_{vib}$ and
$\tau_{rot}$ as given by relation before (70) of~\cite{domenico} (with
equilibrium temperatures), we find exactly $\zeta$.
\end{remark}

Our second result is the Chapman-Enskog distribution for our model.
\begin{proposition}
  The first order expansion of $f$ is
  \begin{equation*}
f=\mathcal{M}[f]-\tau \mbox{Kn}\mathcal{M}[f]\left( A(V,J,K)\cdot\frac{\nabla (RT_{eq})}{\sqrt{RT_{eq}}}+B(V,J,K):\nabla u  \right)  +O(\Kn^2),
  \end{equation*}
  with
  \begin{align*}
& V=\frac{v-u}{\sqrt{RT_{eq}}}, \quad J=\frac{\varepsilon}{RT_{eq}},
                  \quad K=\frac{iT_0}{T_{eq}}, \\
& A=A_{tr}+A_{rot}+A_{vib} =\left( \frac{|V|^2}{2}-\frac{5}{2} \right)V + \left( J-\frac{\delta}{2} \right)V +\left( K-\frac{\delta_v(T_{eq})}{2} \right)V ,\\
& B=B_{tr}+B_{rot}+B_{vib}, \\
& B_{tr}(V)=\Prandtl\left(V \otimes V-\left( \left( \frac{|V|^2}{2}-\frac{3}{2} \right)(\frac23 - \frac\zeta\mu)+1 \right)I\right), \\
& B_{rot}(V,J)=-\left(\frac{\tauc}{\tau}Z_{rot}(\gamma-1) - \frac{\zeta}{\mu}\Prandtl\right)
                \left(J-\frac{\delta}{2} \right)I, \\
& B_{vib}(V,K)=-\left(\frac{\tauc}{\tau}Z_{vib}(\gamma-1) - \frac{\zeta}{\mu}\Prandtl\right)
                \left( K-\frac{\delta_v(T_{eq})}{2} \right)I.
  \end{align*}

   \end{proposition}
This result can be obtained exactly as in~\cite{DMM_2021}.

\section{\label{sec: r_esbgk}Reduced ES-BGK model}

For numerical simulations with a deterministic solver, our ES-BGK
model may be  too expensive, since 
  it depends on many variables: time $t\in\mathbb{R}$, position
$x\in\mathbb{R}^3$, velocity $v\in\mathbb{R}^3$, rotational energy
$\varepsilon\in\mathbb{R}^+$ and discrete levels of the vibrational
energy $i\in\mathbb{N}$. For aerodynamic problems, it is generally 
sufficient to compute the macroscopic velocity and temperatures fields: a
reduced distribution technique~\cite{Chu_1965} (by integration w.r.t rotational
and vibrational energy) permits to drastically reduce the
computational cost, without any approximation {(as long as
  boundary conditions are compatible with this reduction, like usual
  equilibrium inflow boundary conditions and Maxwell reflection at a
  solid wall, for instance)} . We define the three
marginal distributions: 
\begin{equation*}
  \begin{pmatrix}
F(t,x,v) \\ G(t,x,v)\\ H(t,x,v)
 \end{pmatrix}
 = \sum_{i=0}^{+\infty}\int_{\mathbb{R}}
 \begin{pmatrix}
1 \\ \varepsilon \\i R T_0
 \end{pmatrix}
  f(t,x,v,\varepsilon,i) \, d\varepsilon.
\end{equation*}
The macroscopic quantities defined
by~\eqref{eq-mtsf}--\eqref{eq-Theta_q}  now depend on
$F$, $G$ and $H$ through:
\begin{equation}
\begin{aligned}
  & \rho=\la F \ra_{v},\quad \rho u=\la vF \ra_{v},\\
  & \rho E_{tr}(f)=\la\frac{1}{2}|v|^2 F \ra_{v},
  \quad \rho E_{rot}(f)= \la G \ra_{v},\quad \rho E_{vib}(f)=\la H \ra_{v}, \\
& \rho \Theta =\la(v-u)\otimes (v-u)F \ra_{v},\quad q= \la (\frac{1}{2}|v-u|^2 F+G+H)(v-u) \ra_{v},
\end{aligned}
\end{equation}
where $\la . \ra_{v}$ denotes integrals with respect to $v$ only.
The reduced ES-BGK is obtained by multiplying our kinetic model~\eqref{eq: cinetique}-\eqref{eq: operator}
by the vector $(1,\varepsilon,iRT_0)^T$ and by summing and integrating
w.r.t to $i$ and $\varepsilon$, respectively: it is written
\begin{equation}
\label{eq: r_model}
\partial_t \F +v\cdot \nabla \F = \frac{1}{\tau}({\boldsymbol{\mathcal{G}}}[\F]-\F).
\end{equation}
with $\F = (F,G,H)$ and $ {\boldsymbol{\mathcal{G}}}[\F] = ({{\mathcal{G}}}_{tr}[f],\erot(\Trotrel){{\mathcal{G}}}_{tr}[f],\evib(\Tvibrel){{\mathcal{G}}}_{tr}[f])$.

By using the same argument as in~\cite{DMM_2021} and the result of 
proposition~\ref{prop:theoH}, we can prove the following H-theorem for
this reduced model. The proof is left to the reader.
\begin{proposition}
The functional $\entred(\F) = \la h(\F) \ra_v$, where
\begin{equation}\label{eq-hFGH} 
h(\F) = F\left[
  \left(1+\frac{\delta}{2}\right)\log\left(\frac{F}{G^{\frac{\delta}{2+\delta}}}\right)
  +\log\left(\frac{RT_0F}{RT_0F+H}\right)
   \right]  +\frac{H}{RT_0}\log\left(\frac{H}{RT_0F+H}\right)
\end{equation}
is an entropy for the reduced ES-BGK system~\eqref{eq: r_model} and we
have
  \begin{equation}\label{eq-Htheo-reduced} 
    \partial_t \entred(\F) + \nabla \cdot \la v h(\F)\ra_v =
    \la \nabla_{\F}h(\F)\cdot (    \frac{1}{\tau}{\boldsymbol{\mathcal{G}}}[\F] - \F)
    \ra_v
    \leq 0,
  \end{equation}
under conditions of
propositions~\ref{prop:pos_ener},~\ref{pipositive},
and~\ref{prop:theoH}. The equilibrium is
reached (the right-hand side of~\eqref{eq-Htheo-reduced} is zero) if, and
only if,
\begin{equation*}
\F = (\mathcal{M}_{tr}[f],
e_{rot}(T_{eq})\mathcal{M}_{tr}[f], e_{vib}(T_{eq})\mathcal{M}_{tr}[f]),  
  \end{equation*}
where $\mathcal{M}_{\mo}[f]$ is the Maxwellian for translation modes
(see section~\ref{subsec:construct}).
  \end{proposition}

\section{Numerical results} \label{sec: num}
To test the ES-BGK model presented here, we will choose a Monte Carlo approach. For this purpose, relaxation processes to the equilibrium state in an adiabatic box will be investigated and compared with analytical solutions and results of the DSMC method. In the homogeneous test cases, the equation to be solved simplifies to 
\begin{equation}
\label{eq:homESBGK}
\partial_t f=\frac{1}{\tau}({\mathbf{\mathcal{G}}}[f]-f).
\end{equation}

In the Monte Carlo method, the distribution function is represented by a linear combination of $N$ delta functions in phase space with a numerical weighting $w$. The points in the phase space are often interpreted as particles, where $w$ corresponds to the number of real particles that a simulation particle represents. For the Monte Carlo method, \eqref{eq:homESBGK} is integrated analytically for a time step $\Delta t=t^{n+1}-t^n$~\cite{pfeiffer2018particle}:
\begin{equation}
    f^{n+1} = (1-\exp{(-\Delta t/\tau)})\mathbf{\mathcal{G}}[f^n] + \exp{(-\Delta t/\tau)}f^n.
\end{equation}

The idea is that each of the $N$ particles relaxes with probability $(1-\exp{(-\Delta t/\tau)})$, i.e. a new state is sampled from the distribution function $\mathbf{\mathcal{G}}[f^n]$. Different ways to efficiently sample velocities from the ES-BGK distribution are described in~\cite{pfeiffer2018particle}. The new rotational energy is sampled using an exponential distribution depending on $T_{rot}^{rel}$. The new vibrational quantum state is sampled using a standard Acceptance-Rejection method as described in~\cite{bird} depending on $T_{vib}^{rel}$. Since energy and momentum are only preserved in the mean here, we choose a large number of particles $N$ and small time steps $\Delta t/\tau<0.1$. This reduces the statistical noise and ensures stability.

\subsection{Comparison with analytical Results}
To compare the model with an analytical solution, the translational-rotational and translational-vibrational relaxation are first considered separately, i.e. $Z_{vib}=\infty$ and $Z_{rot}=\infty$ respectively. If one also assumes an isothermal relaxation ($T_{tr}(t)=T_{tr}(t=\infty)$), i.e. the thermal velocities do not relax, the characteristic time $\tauc$ is constant (since it depends on the translation temperature) and it is possible to define an analytical solution of the Landau-Teller equation
\begin{equation}
    \frac{E_i(t=\infty)-E_i(t)}{E_i(t=\infty)-E_i(t=0)}=e^{-t/Z_i\tauc}
    \label{eq:analyLT}
\end{equation}
with $i$ being the rotational or vibrational part.

The simulations are done with nitrogen N$_2$ with a characteristic vibrational temperature of $T_0^\mathrm{N_2}=3395\,\mathrm{K}$ using a VHS collision model. This means an exponential ansatz is used for the viscosity depending on the VHS parameters $T^{VHS}_{dref}=273\,\mathrm{K}$, $d^{VHS}_{dref}=4.17\cdot10^{-10}\,\mathrm{m}$ and $\omega_{VHS}=0.74$ as described in \cite{pfeiffer2018}. Thus, the analytical value for $\tauc$ can be calculated with the fixed translational temperature as described in Sec. \ref{subsec:chartime}.

The particle density in the simulations was chosen to be $n=2\cdot10^{22}\,\mathrm{m^{-3}}$ which corresponds to about 4 million particles in our simulation. The translational temperature is fixed to $T_{tr}=T_{eq}=16000\,\mathrm{K}$, the initial temperatures of the rotational and vibrational states are $T_{rot}=T_{vib}=8000\,\mathrm{K}$. The collision numbers are chosen to $Z_{rot}=5$ and $Z_{vib}=10$. The results for the normalized energy difference (left hand side of \eqref{eq:analyLT}) are depicted in Fig.~\ref{fig:LT_relax} showing a very good agreement for the rotational as well as vibrational relaxation.

\begin{figure}
\centering
\subfloat[Rotational relaxation.\label{fig:ltrot}]{\begin{tikzpicture}
\tikzset{
every pin/.style={font=\tiny, thick, fill opacity=0.45,text opacity=1, text=black},
every pin edge/.style={draw=black},
small dot/.style={fill=black,circle,scale=0.25},
}
\begin{axis}[
domain=0:1.5E-6,
width=0.45\textwidth,
height=0.45\textwidth,
xlabel={\small t [s] },
ylabel={\small Normalized energy difference [-]},
ylabel shift=-3 pt,
every axis plot/.append style={thick},
legend style={
	font=\tiny,
	legend cell align=left,
	at={(0.95,0.75)},
	anchor=south east,
			fill=none,
			draw=none}
]
\addplot[color=black, smooth, no markers] {exp(-x*4045190.3662899621 )};
\addplot[color=red, smooth, only marks, mark repeat=15,mark=o, restrict x to domain=0:2.2E-6,mark size =1.3]  table [col sep=comma,x expr=(\thisrow{001-TIME}), 
y expr=(4.40465243514385E-013-\thisrow{005-E-Rot001     })/(4.40465243514385E-013-2.20763473670618E-013)] {LT_Rot.csv};

\legend{Analytic, ESBGK}
\end{axis}
\end{tikzpicture}}
\subfloat[Vibrational relaxation.\label{fig:ltvib}]{\begin{tikzpicture}
\tikzset{
every pin/.style={font=\tiny, thick, fill opacity=0.45,text opacity=1, text=black},
every pin edge/.style={draw=black},
small dot/.style={fill=black,circle,scale=0.25},
}
\begin{axis}[
domain=0:3.5E-6, 
width=0.45\textwidth,
height=0.45\textwidth,
xlabel={\small t [s] },
ylabel={\small Normalized energy difference [-]},
ylabel shift=-3 pt,
every axis plot/.append style={thick},
legend style={
	font=\tiny,
	legend cell align=left,
	at={(0.95,0.7)},
	anchor=south east,
			fill=none,
			draw=none}
]
\addplot[color=black, smooth, no markers] {exp(-x*2022595.1831450218)};
\addplot[color=red, smooth, only marks, mark repeat=25, mark =o, restrict x to domain=0:3.5E-6,,mark size =1.3]  table [col sep=comma,x expr=(\thisrow{001-TIME}), y expr=(4.42958911051589E-013-\thisrow{004-E-Vib001     })/(4.42958911051589E-013-2.24001935183863E-013)] {LT_Vib.csv};

\legend{Analytic, ESBGK}
\end{axis}
\end{tikzpicture}}
\caption{Comparison of ES-BGK simulation results with analytical Landau-Teller solution.}
\label{fig:LT_relax}
\end{figure}
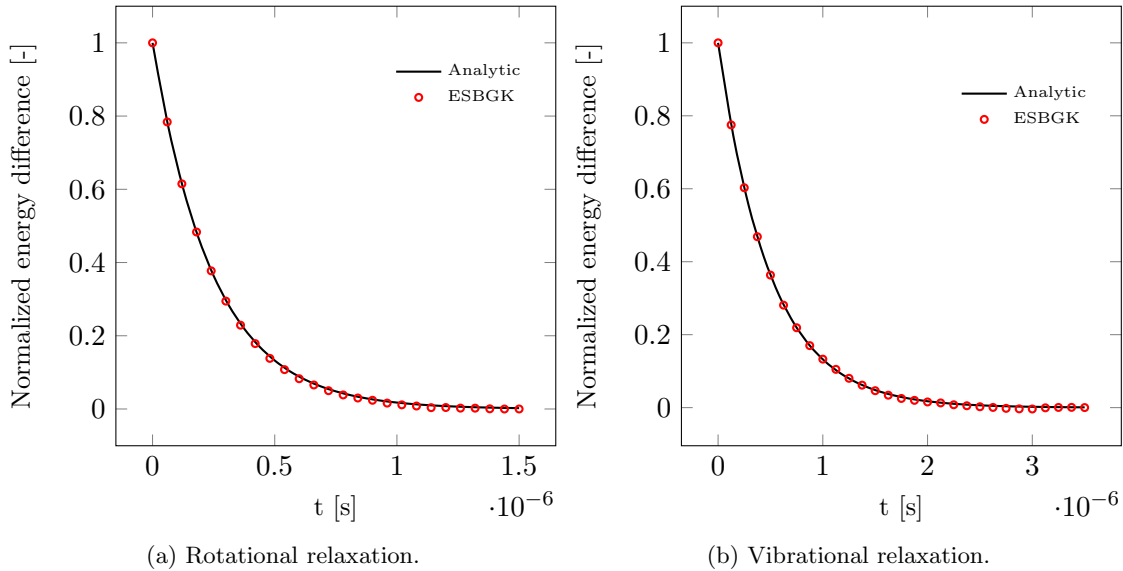

\subsection{Comparison with DSMC Results}
In this test case, a simultaneous relaxation of the translational, rotational and vibrational temperature is demonstrated and compared with DSMC. Furthermore, the difference is shown when it is assumed that there is only one relaxation time, i.e. $\tauc=\tau$. For this simulation the same parameters have been used as before with the exception of $T_{tr}=16000\,\mathrm{K}$, $T_{vib}=8000\,\mathrm{K}$, $T_{rot}=12000\,\mathrm{K}$, $Z_{rot}=5$ and $Z_{vib}=50$. The DSMC simulation was carried out with the identical VHS parameters. In addition, the prohibiting double relaxation method was used to reproduce the Landau-Teller equation as described in various studies~\cite{zhang:2013,Pfeiffer2016b}.
The results are depicted in Fig.~\ref{fig:compdsmcbgk} and excellent agreement is found between DSMC and the proposed ES-BGK model. Furthermore, it is easy to see that the model with only one relaxation time does not produce the correct Landau-Teller relaxation curves when the same $Z_{rot}$ and $Z_{vib}$ are used: this clearly proves the improvement of our new model.
\begin{figure}
\centering
\begin{tikzpicture}
\begin{axis}[
width=0.8\textwidth,
height=0.6\textwidth,
xlabel={\small t [s] },
ylabel={\small Temperature [K]},
scaled x ticks = true,
ylabel shift=-3 pt,
every axis plot/.append style={thick},
legend style={
  legend columns = 3,
	font=\tiny,
	legend cell align=left,
	at={(0.98,0.725)},
	anchor=south east,
			fill=none,
			draw=none}
]
\addplot[color=black, smooth, restrict x to domain=0:1.35E-5]  table [col sep=comma,x expr=(\thisrow{001-TIME}), y expr=(\thisrow{002-TempTra-001     })] {DSMC_full.csv};
\addplot[color=red, smooth, only marks, mark =o, mark repeat=55, restrict x to domain=0:1.35E-5,mark size =1.3]  table [col sep=comma,x expr=(\thisrow{001-TIME}), y expr=(\thisrow{002-TempTra-001     })] {BGK_full.csv};
\addplot[color=blue, smooth, restrict x to domain=0:1.35E-5]  table [col sep=comma,x expr=(\thisrow{001-TIME}), y expr=(\thisrow{002-TempTra-001     })] {BGK_full_wnu.csv};
\addplot[color=black,dotted, smooth, restrict x to domain=0:1.35E-5]  table [col sep=comma,x expr=(\thisrow{001-TIME}), y expr=(\thisrow{003-TempVib001     })] {DSMC_full.csv};
\addplot[color=red, smooth, only marks, mark =x, mark repeat=50, restrict x to domain=0:1.35E-5,mark size =1.3]  table [col sep=comma,x expr=(\thisrow{001-TIME}), y expr=(\thisrow{003-TempVib001     })] {BGK_full.csv};
\addplot[color=blue,dotted, smooth, restrict x to domain=0:1.35E-5]  table [col sep=comma,x expr=(\thisrow{001-TIME}), y expr=(\thisrow{003-TempVib001     })] {BGK_full_wnu.csv};
\addplot[color=black,dashed, smooth, restrict x to domain=0:1.35E-5]  table [col sep=comma,x expr=(\thisrow{001-TIME}), y expr=(\thisrow{005-TempRot001     })] {DSMC_full.csv};
\addplot[color=red, smooth, only marks, mark repeat=55, mark=square, restrict x to domain=0:1.35E-5,mark size =1.3]  table [col sep=comma,x expr=(\thisrow{001-TIME}), y expr=(\thisrow{005-TempRot001     })] {BGK_full.csv};
\addplot[color=blue,dashed, smooth, restrict x to domain=0:1.35E-5]  table [col sep=comma,x expr=(\thisrow{001-TIME}), y expr=(\thisrow{005-TempRot001     })] {BGK_full_wnu.csv};


\legend{DSMC: $\Ttr$\qquad { }, ESBGK: $\Ttr$\qquad { },  ESBGK  with $\tauc=\tau$: $\Ttr$, DSMC: $\Tvib$ \qquad { } , 
ESBGK: $\Tvib$\qquad { }, ESBGK with $\tauc=\tau$: $\Tvib$, DSMC: $\Trot$ \qquad { },  ESBGK: $\Trot$\qquad { } ,  ESBGK with $\tauc=\tau$: $\Trot$}

\end{axis}
\end{tikzpicture}
\caption{Comparison of relaxation process of $T_{tr}$, $T_{rot}$ and $T_{vib}$ between DSMC
and ES-BGK as well as ES-BGK with only one relaxation time $\tauc=\tau$.}
\label{fig:compdsmcbgk}
\end{figure}
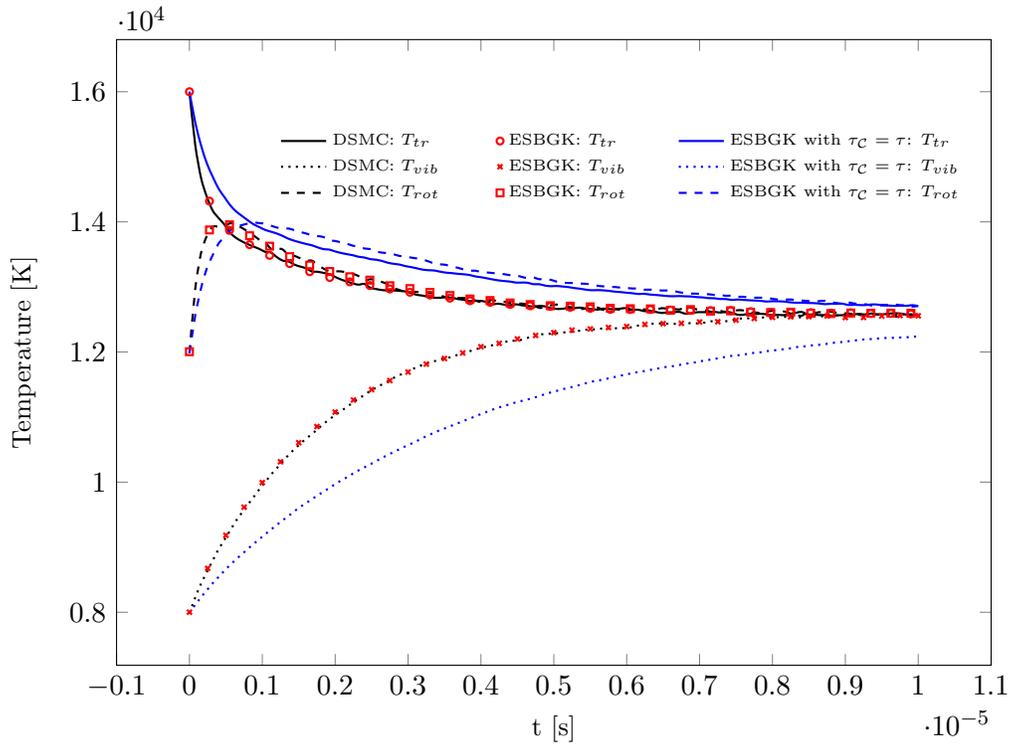

\section{Conclusion}

In this paper, we have proposed an ES-BGK model for diatomic gases
that accounts for translational-rotational and
translational-vibrational energy exchanges. 
It is consistent with the general definition of the vibrational and rotational collision numbers 
that are also commonly used in DSMC solvers to reproduce the Landau-Teller
and Jean equations.

Our model is based on a correction of a previous model~\cite{DMM_2021}
and is induced by the numerical method of~\cite{pfeiffer2018}. We have
proved this model satisfies the H-theorem and fits the correct
transport coefficients. Even the volume viscosity is consistent
with that obtained for a Boltzmann equation with two internal energy
modes in~\cite{domenico}.

In the purely translational-rotational case, our model also gives a
correction of the standard ES-BGK model of Andriès et al.~\cite{ALPP}
with a correction factor of the collision number that can be as large
as $3$.

A reduced version of our model has been derived to eliminate the
dependency to the internal energy variables: the reduced model should
make it possible numerical simulations of diatomic gas flows with a
computational cost of same order of magnitude as for a
monoatomic gas.

This model will be extended to polyatomic molecules with more than two atoms in a forthcoming work.

\appendix

\section{Convexity of the vibrational energy}
\label{app:evibconvex}

Differentiation of~\eqref{eq-icvib} with respect to $T$ gives
\begin{equation*}
  \frac{d}{dT}\cvvib(T) = \frac{2RT_0^2}{T^3e^{T_0/T}(e^{T_0/T}-1)^2}
  \left(\frac{T_0}{2T}\coth(\frac{T_0}{2T})  -1\right),
\end{equation*}
which is positive, from the known inequality $\coth(x)\geq 1/x$ for
every positive $x$. This proves that $\cvvib$ is an increasing
function, and hence that $\evib$ is convex. Passing to the limit $T=0$
in~\eqref{eq-icvib} shows that $\cvvib$ is bounded
by $R$.

We come to the proof of the assertion made in first case of step 3 for the
proof of proposition~\ref{prop:theoH}. We have $\Ttr\leq T_1\leq
\Tvib$ and $\Tvibrel\leq T_2\leq \Tvib$, while $\Tvibrel$ is between
$\Ttr$ and $\Tvib$, and we want to prove that $\cvvib(T_2)\geq
\cvvib(T_1)$. This is a simple consequence of the convexity of
$\evib$, as it is shown below.

It is well known that for for every convex function $\phi$, the ratio
$(\phi(y)-\phi(x))/(y-x)$ is increasing in $x$ for every fixed
$y$. This result applied to $\phi=\evib$, $y=\Tvib$ and $x=\Ttr$ then $x=\Tvibrel$ gives
\begin{equation*}
  (\evib(\Tvib)-\evib(\Tvibrel))/(\Tvib-\Tvibrel)
  \geq
    (\evib(\Tvib)-\evib(\Ttr))/(\Tvib-\Ttr).
  \end{equation*}
Then we remind that $T_1$ and $T_2$ are defined by~\eqref{eq-defT1T2}:
the previous inequality is then exactly $\cvvib(T_2)\geq
\cvvib(T_1)$.

\section{Inequality for $\det(\Theta)/det(\Pi)$}
\label{subsec:ineqthetapi}
This proof is very close to that given in~\cite{ALPP}. First, note
that~\eqref{ineqthetapi} is equivalent to
\begin{equation}  \label{eq-inegthetapi}
\det(\frac{\Theta}{R\Ttr}) \leq \det (\frac{\Pi}{R\Ttrrel}).
  \end{equation}
Then as remarked in the proof of proposition~\ref{prop:pos_ener}, we can
work in the same basis in which $\Theta$ and $\Pi$ are diagonal
tensors, and we denote by $\mu_i$ the three positive eigenvalues
of $\Theta/R\Ttr$, whose sum is
$3$. Then by using definition~\eqref{eq-defPi} of $\Pi$,~\eqref{eq-inegthetapi} reads
\begin{equation}  \label{eq-inegthetapi2}
\prod_{i=1}^3 \mu_i \leq \prod_{i=1}^3 (1 + \alpha (\mu_i - 1)),
\end{equation}
where we set $\alpha  = \frac{\Prandtl -
  1}{\Prandtl}\frac{\Ttr}{\Ttrrel}$. We remind that $\Pi$ is positive
definite under assumption~\eqref{condpi}, so that each terms in the
product of the right-hand side of~(\ref{eq-inegthetapi2}) is
positive. Therefore, we can apply the $\log$ function to this
inequality to get
\begin{equation}  \label{eq-inegthetapi3}
\sum_{i=1}^3 \log \mu_i \leq \sum_{i=1}^3 \log(1 + \alpha (\mu_i - 1)).
\end{equation}
This is the inequality we prove now.

As usual, the idea is to use convexity properties. However, since
$\alpha$ is negative for $\Prandtl$
between 2/3 and 1, the right-hand side of~(\ref{eq-inegthetapi3})
is rewritten by using $(\mu_1 +\mu_2 +\mu_3)/3 = 1$. Indeed, we get
\begin{equation}\label{eq-inegthetapi4}
  \begin{split}
    \sum_{i=1}^3 \log(1 + \alpha (\mu_i - 1))
   &  = \sum_{i=1}^3 \log( \frac13(\mu_1 +\mu_2 +\mu_3) + \alpha (\mu_i
    - \frac13(\mu_1 +\mu_2 +\mu_3)  )) \\
   & = \sum_{i=1}^3 \log( \frac13 (1-\alpha)(\mu_{i_1} + \mu_{i_2}) +
   \frac13(1+2\alpha)\mu_i ),
\end{split}
\end{equation}
where $i_1$ and $i_2$ are the indices that follow $i$ in the circular
permutation of $\lbrace 1,2,3\rbrace$.
Now we assume $\alpha>-1/2$ (see below), so that the argument of the
$\log$ function above is a convex combination of the $\mu_i$. Since $\log$ is concave, the
Jensen inequality gives
\begin{equation}\label{eq-inegthetapi5}
  \begin{split}
\sum_{i=1}^3 \log( \frac13 (1-\alpha)(\mu_{i_1} + \mu_{i_2}) +
\frac13(1+2\alpha)\mu_i ) & \geq
\sum_{i=1}^3 \frac13 (1-\alpha)(\log(\mu_{i_1}) + \log(\mu_{i_2})) +
\frac13(1+2\alpha)\log (\mu_i ) \\
& = \sum_{i=1}^3 \log \mu_i = \det (\frac{\Theta}{R\Ttr}).
  \end{split}
\end{equation}
This proves~(\ref{eq-inegthetapi3}).

It remains to prove $\alpha>-1/2$: this is actually a consequence of
assumption~\eqref{condpi} that garantees that $\Pi$ is positive
definite (see the proof of proposition~\ref{pipositive}: the positivity
of $\alpha$ is equivalent to that of the right-hand side of
\eqref{eq-minlambdaiPi}, which is given by~\eqref{eq-cond2}, and hence by~\eqref{condpi}).

\section{Elements of proof for the Chapman-Enskog expansion}
\label{sec:hydrodynamic-limits}

Integration of~\eqref{eq: nd_kinetic} multiplied by $\demi |v|^2$,
$\varepsilon$, and $iRT_0$, respectively, gives macroscopic evolution
equations of $\etr(\Ttr)$, $\erot(\Trot)$, and
$\evib(\Tvib)$. Linearization of these equations by using $\partial_t
e_{\alpha}(T_{\alpha}) = \cv^{\alpha}(T_{\alpha})\partial_tT_{\alpha}$
give first order expansions of $\etr(\Ttrrel)$, $\erot(\Trotrel)$, and
$\evib(\Tvibrel)$. The definition of the relaxation
energies~\eqref{eq-Trotrel}--\eqref{eq-Ttrrel2} and other successive
linearizations lead to the
following first order expansions:
\begin{align*}
  & T_{rot}^{rel}
    =T_{eq}\left(1-\Kn\tauc (\gamma-1)\left(
                   Z_{rot}\left(\frac{\cvrot}{\cv(T_{eq})}-1\right)
                   + Z_{vib}\frac{\cvvib(T_{eq})}{\cv(T_{eq})}
    + \frac{\tau}{\tauc}\right)
    \nabla\cdot u\right)+O(\Kn^2),\\
  &  T_{vib}^{rel}=T_{eq}\left(1-\Kn\tauc (\gamma -1)\left(
    Z_{vib}\left(\frac{\cvvib(T_{eq})}{\cv(T_{eq})}-1\right)
    +Z_{rot}\frac{\cvrot}{\cv(T_{eq})}
    +\frac{\tau}{\tauc}\right)
    \nabla\cdot u\right)+O(\Kn^2),\\
  &  T_{tr}^{rel}=T_{eq}\left(1-\Kn\tauc (\gamma -1)\left(
    Z_{vib}\frac{\cvvib(T_{eq})}{\cv(T_{eq})}
    +Z_{rot}\frac{\cvrot}{\cv(T_{eq})}
    +\frac{\tau}{\tauc}\left(1-\frac{\cv(T_{eq})}{\cvtr}\right)
    \right)\nabla\cdot u\right)+O(\Kn^2).
\end{align*}
These relations give the first order expansion of the relaxation
tensor
\begin{equation*}
  \begin{split}
\Pi=&\frac{RT_{eq}}{\Prandtl}I+\frac{\Prandtl-1}{\Prandtl}\Theta\\
+&\Kn\tau R(\gamma-1)T_{eq}
\left(
  \left(\frac{\cv(T_{eq})}{\cvtr}-1\right)
  -\frac{1}{\Prandtl}\frac{\tauc}{\tau}
      \left(Z_{vib}\frac{\cvvib(T_{eq})}{\cv(T_{eq})}
            +Z_{rot}\frac{\cvrot}{\cv(T_{eq})}\right)\right) \nabla\cdot u I\\&+O(\Kn^2).
\end{split}
\end{equation*}
The other calculations are standard and can be found in~\cite{DMM_2021}.

\section*{Acknowledgments}

"Marcel Pfeiffer has received funding from the European Research Council (ERC) under the European Union’s Horizon 2020 research and innovation programme (grant agreement No. 899981 MEDUSA)"


\bibliographystyle{authordate1}
\bibliography{biblio}


\end{document}